\documentclass{article}
\usepackage{hyperref} 
\usepackage{Definitions}

\usepackage{fullpage}

\usepackage[numbers]{natbib}

\usepackage{authblk}
\author[1]{Ameet Gadekar\footnote{\url{ameet.gadekar@cispa.de}}}
\author[2]{Aristides Gionis\footnote{\url{argioni@kth.se}}}
\author[3]{Suhas Thejaswi\footnote{\url{firstname.lastname@aalto.fi}}}
\author[2]{Sijing Tu\footnote{\url{sijingtu93@gmail.com}}}
\affil[1]{CISPA Helmholtz Center for Information Security, Germany}
\affil[2]{KTH Royal Institute of Technology, Sweden}
\affil[3]{Aalto University, Finland}

\usepackage{lineno}

\usepackage[dvipsnames]{xcolor}

\newcommand{\bigO}{\ensuremath{\mathcal{O}}\xspace}

\newcommand{\BibTeX}{\rm B\kern-.05em{\sc i\kern-.025em b}\kern-.08em\TeX}

\date{}
\begin{document}

\title{Fair Committee Selection under Ordinal Preferences and Limited Cardinal Information}
\maketitle
\begin{abstract}

We study the problem of fair $k$-committee selection under an egalitarian objective. Given $n$ agents partitioned into $m$ groups (\eg, demographic quotas), the goal is to aggregate their preferences to form a committee of size $k$ that guarantees minimum representation from each group while minimizing the maximum \emph{cost} incurred by any agent.
We model this setting as the ordinal fair $k$-center problem, where agents are embedded in an unknown metric space, and each agent reports a complete preference ranking (\ie, ordinal information) over all agents, consistent with the underlying distance metric (\ie, cardinal information).
The cost incurred by an agent with respect to a committee is defined as its distance to the closest committee member.
The quality of an algorithm is evaluated using the notion of distortion, which measures the worst-case ratio between the the cost of the committee produced by the algorithm and  the cost of an optimal committee, when given complete access to the underlying metric space.

When cardinal information is not available, no constant distortion is possible for the ordinal $k$-center problem, even without fairness constraints,  when $k\geq 3$ [Burkhardt et.al., AAAI'24]. 
To overcome this hardness, we allow limited access to cardinal information by querying the metric space. In this setting, our main contribution is a factor-$5$ distortion algorithm that requires only $\bigO(k \log^2 k)$ queries. Along the way, we present an improved factor-$3$ distortion algorithm using 
$\bigO(k^2)$ 
queries.

\end{abstract} 

\def\thefootnote{*}\footnotetext{Authors are listed in alphabetical order. During this work, Suhas Thejaswi was employed at the Max Planck Institute for Software Systems, Germany, and part of this research was done while Suhas Thejaswi and Ameet Gadekar were visiting KTH Royal Institute of Technology, Sweden.}\def\thefootnote{\arabic{footnote}}

\section{Introduction}
In many collective decision-making tasks, a group of agents (\eg, voters) needs to select a subset of alternatives (\eg, candidates) that provides equitable outcomes for each agent while respecting additional demographic fairness constraints.
We consider this problem in the context of committee selection, which generalizes the single-winner selection problem to selecting multiple representatives.
Specifically, we assume that agents and alternatives are embedded in a metric space, where the distance between an agent and an alternative represents the cost (or dissatisfaction) that the agent experiences when that alternative is selected. 
Each agent can provide a complete ranking over the alternatives that is consistent with these underlying distances: if an agent ranks alternative $a$ above alternative $b$, then the distance from the agent to $a$ is no greater than the distance from the agent to $b$. 
The computational social choice literature refers to these preference rankings as ordinal information, as opposed to the explicit distance values, which are referred to as cardinal information. 
The goal is to select a committee (subset of alternatives) that minimizes the social cost, defined as a function of the distances of agents to their closest member in the selected committee. 

While ordinal rankings are readily available, the exact distance values are typically unknown to the algorithm.
A common way to evaluate the performance of a social choice algorithm that uses only ordinal information is through the notion of \emph{distortion}, first introduced by \citet{procaccia2006distortion}.
This concept was later extended to metric preferences by~\cite{anshelevich2018approximating}, which measures the quality of the solution produced by the algorithm, defined as the worst-case ratio between the social cost of the algorithm's solution and the social cost of the optimal solution that has full access to all pairwise distances.
In recent years, metric distortion has been a central tool for evaluating the performance of social choice algorithms in the single-winner setting~\cite{anshelevich2018approximating, munagala2019improved, gkatzelis2020resolving, kizilkaya2022plurality, kizilkaya2023generalized, 
charikar2022metric, fain2019random,
charikar2024breaking}, \ie, when the committee consists of a single member.

Recently, \citet{pulyassary2025constant} and \citet{burkhardt2024low} extend the metric distortion framework from single-winner selection to $k$-committee selection, for arbitrary $k$. 
Their works consider both the utilitarian objective, which minimizes the sum of distances of each agent to its nearest committee member (which corresponds to the $k$-median or $k$-means objective), and the egalitarian objective, which minimizes the maximum distance of an agent to its nearest committee member (which corresponds to the $k$-center objective).
\citet{burkhardt2024low} show that for $k \geq 3$, no algorithm can achieve a constant distortion for the ordinal $k$-committee selection problem with the egalitarian objective (i..e, when cardinal information is not available).
In light of this hardness result, both works additionally allow limited access to cardinal information through distance queries in order to obtain constant-factor distortion guarantees.

In this paper, we continue this line of work by studying a fair version of $k$-committee selection. 
Specifically, we study the problem of \emph{ordinal fair $k$-committee selection} with an emphasis on both \emph{equity} and \emph{fairness}. 
The equity principle ensures that no agent is unduly disadvantaged by the collective outcome; we encode this through the egalitarian objective, which minimizes the maximum distance from any agent to its nearest committee member.
To enforce fairness in the committee, we assume that the alternatives are partitioned according to an underlying demographic profile. The goal is to guarantee minimum representation from each demographic group in the selected committee. Numerous studies in social welfare highlight that such minimum representation can promote equity, inclusion, and informed policy decisions. For instance, quotas have been widely considered as mechanisms to secure diverse representation and balance opportunities across groups~\cite{barnes2020gender, croissant2019ethnic}. Our fairness constraints are designed to align with this notion of demographic fairness. Additionally, similar to~\cite{burkhardt2024low,cembrano2025metric}, we focus on the setting where alternatives and agents coincide.

Informally, the ordinal fair $k$-committee selection problem is defined as follows. 
We are given a set of agents embedded in an unknown metric space, partitioned into demographic groups. 
Each agent provides a complete ranking over all other agents that is consistent with the underlying distances in the metric space.
The goal is to select a committee of size $k$ using only a limited number of distance queries, subject to fairness constraints that ensure minimum representation from each group, while minimizing the maximum distance from any agent to its nearest committee member (the egalitarian objective).
The key difference between our work and \citet{pulyassary2025constant} and \citet{burkhardt2024low} is that we incorporate fairness constraints on minimum representation requirements for each demographic group.

When full cardinal information is available, the ordinal fair $k$-committee selection problem is equivalent to the fair $k$-center problem~\cite{kleindessner2019fair}. 
To develop algorithmic solutions for the ordinal fair $k$-committee selection problem, we build on techniques from the fair $k$-center problem. 
For consistency with existing literature, we formulate our problem as the ordinal fair $k$-center problem with limited cardinal information, providing a natural extension to clustering methodologies that maintains alignment with prior work. 
We give a formal definition of the problem in Definition~\ref{def:ordinal-fair-k-center-with-limited-distance-queries}.

\xhdr{Our contributions} 
In detail, our contributions are as follows.

$\bullet$ We initiate the study of low-distortion mechanisms for committee selection that incorporate \emph{fairness} constraints under the egalitarian (min-max) objective for the social cost, by modeling the task as the ordinal fair $k$-center problem with limited access to cardinal information. 

$\bullet$ We present a $5$-distortion algorithm with only $\bigO(k \log ^2 k)$ cardinal queries. 
Additionally, we present a $3$-distortion algorithm that requires $2k^2$ queries. 

\subsection{Our techniques}
Our methods build upon techniques from both the ordinal setting with limited query access and the clustering domain.
The ordinal $k$-center under limited cardinal information was studied by
\citet{burkhardt2024low}, where they presented two algorithmic results: a $2$-distortion algorithm with $\frac{k^2-k}{2}$ queries and a $4$-distortion algorithm with only $2k$ queries. Both these algorithms are based on adaptions of~\citet{gonzalez1985clustering}'s algorithm.
On the other hand, in the clustering domain,
a common framework that has been successful in solving the fair $k$-center problem (with complete cardinal information), is that of maximal matching~\cite{jones2020fair,gadekar2025fair}.
This approach, on a high level,  first solves the standard $k$-center problem to obtain a good quality set of ``fairness-oblivious'' centers $T$ of size $k$ with good approximation guarantee, and
projects it onto a set of fair centers via maximal matching, inuring only a small loss in the quality.
In our work, we take the set $T$ to be the output of the algorithm of~\cite{burkhardt2024low}, using either the $2$-distortion or $4$-distortion variant. We first describe the underlying idea for obtaining a fair set of centers from $T$ via the maximum matching framework of~\cite{jones2020fair,gadekar2025fair}. We then discuss the main challenges in adapting this framework to the ordinal setting, and finally present our novel techniques for overcoming these challenges.
For clarity of presentation, we restrict our attention to the case with $k$ groups, $\mathbb{G}=\{G_1,\dots,G_k\}$, each having a unit requirement, and the optimal cost is $1$. 
Let $d$ be the distance function of the underlying metric space and $n$ be the number of points (or agents) in the input instance.

The key idea of the matching-based approach consists of two components. 
The first component is to use the properties of the (ordered) set $T=(t_1,\dots,t_k)$ adapted from the approaches of~\citet{gonzalez1985clustering}.  
In particular, such set $T$ has a \emph{critical} index $\ell^* \in [k]$, which is the largest index $\ell \in [k]$ such that the points $T_\ell := (t_1,\dots,t_\ell)$ of $T$ are from different clusters of an (fixed) optimal solution, and the cost of $T_\ell$ is within $\alpha$ factor of the optimal cost, for some constant $\alpha>1$. 
The second component is to project the ``fairness-oblivious'' centers $T_{\ell}$ onto a set of fair centers via maximal matching. 
In more detail, for $\ell \in [k]$, consider the bipartite graph $H_\ell$ between the centers of $T_{\ell}$ (left partition) and the demographic groups $\GroupSet$ (right partition), with edges between $t \in T_{\ell}$ and group $G \in \GroupSet$ if $d(t,G) \le \lambda$, for some non-negative real $\lambda$.
Let $\lambda_{\ell}$ be the minimum $\lambda$ such that there exists a left-perfect matching on $H$, noting that $\lambda_\ell\le 1$.
This implies that, if we find a left-perfect matching on $H_\ell$, then we can obtain a set of fair centers of size $k$ from $T_{\ell}$.
Furthermore, the cost of such a solution, denoted as $Sol(\ell)$, is at most $\cost(T_{\ell}) + \lambda_{\ell}$, by triangle inequality.
Moreover, when $\ell = \ell^*$, $H_{\ell^*}$ has a perfect matching on the left-partition ($T_{\ell^*}$), and hence $T_{\ell^*}$ can be extended to a fair solution with cost $\alpha+1$.

Notice that, $\lambda_{\ell^*}$ belongs to one of the $nk$ distances, and $\ell^{*}$ can be guessed over the range $[k]$.
This immediately implies an algorithm to obtain distortion $3$ (and distortion $5$ resp.), by combining the $2$-distortion (and $4$-distortion resp.) algorithm of \cite{burkhardt2024low} for $T$ with the matching frameworks of~\cite{jones2020fair,gadekar2025fair}, using additional $nk(k\log (nk))$ distance queries, where the $k \log (nk)$ factor is for constructing the (edges of the) bipartite graph. 
This naturally leads to the following question: \emph{can we achieve the same distortion guarantees with smaller number of queries, preferably near-linear in $k$, to the cardinal information?}

\vspace{1em}

For each $\ell \in [k]$, one possible idea is to employ binary search on $nk$ distances to find $\lambda_{\ell}$. However, the number of queries is still linear in $nk$, due to the need to sort these distances for the binary search.
To reduce the number of queries, our first idea is to shrink the search space of candidate distances from $nk$ to $k^2$. For each center $t \in T$ and each group $G$ (there are at most $k$ groups), we identify the nearest point in $G$ to $t$ using only the ordinal ranking of $t$. 
This results in at most $k^2$ distances, and hence requires $k^2$ many queries to construct the bipartite graphs that facilitate us find the minimum $\lambda_{\ell}$ for any $\ell \in [k]$.
This yields a $3$-distortion algorithm using $\bigO(k^2)$ queries.
We give a detailed analysis in Section~\ref{sec:conventional-algorithms}.

A natural way to reduce the distance queries is to conduct a binary search on $\ell$ over function $Sol(\ell)$. However, this not possible:  as $cost(T_{\ell})$ is non-increasing and $\lambda_{\ell}$ is non-decreasing, $Sol(\ell)$ is not necessarily monotone. 
Our first key contribution is introducing the predicate $$P(\ell) \equiv (4 \lambda_{\ell} \leq cost(T_{\ell}));$$
which \emph{is monotone} and facilitates binary search. We find the largest $\ell$, denoted $L$, such that $P(L)$ is true but $P(L+1)$ is false. 
Using properties we derive in Lemma~\ref{lemma:monotonicity-of-lambda-and-cost-init-sol-set}, we prove that $\min\{Sol(L), Sol(L+1)\} \leq 5 cost(S^*)$ where $S^*$ is the optimal solution. This establishes the correctness of our binary-search procedure. 

Next, we handle the distance query complexities of evaluating $P(\ell)$ and $\lambda_{\ell}$, which are completely new problems. 

For the query complexity of computing $\lambda_{\ell}$ (detailed in Section~\ref{section:finding-minimum-lambda}), at a high level, we reduce the search space for $\lambda_{\ell}$ (of size $k^2$) by a factor of $\frac{3}{4}$ at each iteration. To accomplish this, at each round, we identify a \emph{pivot point} in the search space such that at least $\frac{1}{4}$ of the candidate distances are no smaller than the pivot and at least $\frac{1}{4}$ are no greater than the pivot such that at least one part does not contain $\lambda_\ell$. Finally, using a median-of-medians-style subroutine (Algorithm~\ref{alg:finding-lambda-ell} and Algorithm~\ref{alg:median-median}), we find such a pivot; this yields query complexity $O(\ell \log^{2} k)$.
Notice that, for the whole algorithm, we only need to compute $\lambda_{\ell}$ for a fixed $\ell$ twice. Therefore, in the worst case scenario when $\ell = k$, the total query complexity of computing $\lambda_{\ell}$ is $O(k \log^2 k)$.

For the query complexity of evaluating $P(\ell)$, a natural approach is to compute  $\lambda_{\ell}$ and compare it with $T_{\ell}$; however, this approach needs to compute $\lambda_{\ell}$ $\log(k)$ times and leads to  $O(\ell \log^3(k))$ queries. To avoid repetitively computing $\lambda_{\ell}$, we apply the equivalent formulation of $P(\ell)$ as 
$$P(\ell)\equiv (\lambda_{\ell} \leq \frac{cost(T_{\ell})}{4}),$$
which holds precisely when there exists a left-perfect matching when setting $\lambda = \frac{cost(T_{\ell})}{4}$. Using this design, we obtain an overall query complexity of $O(k \log^2 k)$.

\subsection{Further related work}
Our work builds on prior research in multi-winner elections, metric distortion, and (fair) clustering. As these areas have been extensively studied, we only discuss works that are most related to our 
methodological approach.

Besides extending the framework of metric distortion and relating the problem with clustering~\cite{burkhardt2024low, pulyassary2025constant, caragiannis2022metric}, 
other lines of works include the study of the \emph{(approximately) stable committee selection} problem~\cite{fain2018fair,jiang2020approximately}, and selecting a committee of minimum size that is a \emph{Condorcet winner}~\cite{charikar2025six}.
Concurrently, there is a work which studies a \emph{peer selection} problem, where the voters and candidates coincide~\cite{cembrano2025metric}. We note from the technical perspective, they focus on the line metric without distance queries, and their techniques may not directly extend to our setting.

Clustering is a fundamental unsupervised machine learning task that has been extensively studied~\cite{jain1999data}. In recent years, awareness about automated decision-making propagating biases has led to an increase in attention towards algorithmic fairness principles. Consequently, several classic unsupervised learning problems, including clustering, have been reintroduced with fairness constraints~\cite{thejaswi2021diversity, thejaswi2022clustering,thejaswi2024diversity,gadekar2025capacitated,matakos2024fair,zhang2024parameterized,abbasi2023parameterized,gadekar2025fair,chierichetti2017fair}. Among these formulations, the one most relevant to our work is the fair $k$-center problem introduced by \citet{kleindessner2019fair}, which seeks to choose a specified number of representatives from each demographic group while minimizing the egalitarian (min-max) objective, \ie, minimizing the maximum distance from any client to its nearest representative. \citet{jones2020fair} presented a $3$-approximation algorithm using a matching framework. This framework has been extended to solve several fair clustering variants~\cite{gadekar2025fair,chen2024approximation}.

The remainder of the paper is organized as follows. 
Section~\ref{sec:problem-formulation} introduces the necessary terminology and problem definitions. 
Section~\ref{sec:conventional-algorithms} provides an overview of the algorithmic techniques that form the basis for our algorithmic methods and proposes a $3$-distortion algorithm with $\bigO(k^2)$ queries.
Section~\ref{sec:analysis-ordinal-fair-k-center-with-limited-distance-queries} presents our main algorithmic contribution, a $5$-distortion algorithm with $\bigO(k\log^2 k)$ queries.

\section{Problem Formulation}
\label{sec:problem-formulation}
In what follows, we use the term agent(s)---common in computational social choice literature---interchangeably with point(s)---as used in the clustering literature. 
Let $\Mcal$ denote the set of all metric spaces on finite points.
Let $(U, d) \in \Mcal$ be a metric space with distance function $d: U \times U \rightarrow \RR_{\ge 0}$.
For a subset $S \subseteq U$ of points, we use $d(u,\SolSet)$ to denote $\min_{s \in \SolSet} d(u, s)$. 
We consider the setting where ordinal rankings of all points of $U$ are available. 
Specifically, for each point $v \in U$, a linear order $\succ_v: [n] \rightarrow U$ is known. We write $u \succ_{v} u'$ to indicate that $u$ is closer $v$ than $u'$ in the ranking $\succ_v$. 
Furthermore, we assume that the linear orders are \emph{consistent} with $d$,\ie, for every $v \in U$ and for all $u, u' \in U$, $u \succ_v u'$ implies $d(u, v) \leq d(u', v)$. 
An \emph{ordinal profile} of $U$ is a collection of linear orders of all the points of $U$, which we denote by $\succ_U$.
We say an ordinal profile $\succ_U$ of $U$ is \emph{consistent} with $d$ if all the linear orders in $\succ_U$ are consistent with $d$.
Let $\profile(d)$ denote the set of all ordinal profiles of $U$ that are consistent with $d$.
We first introduce the ordinal fair $k$-center problem and make it precise how social cost is defined in this context.

\begin{definition}[The ordinal fair $k$-center problem (\ordfairkcen)]
\label{def:ordinal-fair-k-center-with-limited-distance-queries}
An instance of the ordinal fair $k$-center problem is defined on a set $U$ of $n$ points from a metric space $(U,d)$ with unknown $d$, an integer $k \geq  1$, a collection $\succ_U=\{\succ_u\}_{u \in U}$ of linear orders that is  consistent with $d$, a collection $\GroupSet= \{\Group_1, \dots, \Group_t : \Group_i \subseteq U\}$ of $t$ subsets of data points 
that form a partition of $U$,
and a vectors of requirements $\alphavec = \{\alpha_1,\dots,\alpha_t\}$, where $\alpha_i \ge 0$ corresponds to the requirement of group $\Group_i$. A set $\SolSet \subseteq U$ of centers is a feasible solution if $|\SolSet| = k$ and $\alpha_i \leq |\SolSet \cap \Group_i|$ for all $i \in [t]$. 
The \emph{social cost} of a solution $\SolSet$ is the maximum distance of any point to $S$, \ie, $\cost(S) = \max_{u \in U} d(u,\SolSet)$. The goal of the ordinal fair $k$-center problem is to find a feasible solution with minimum social cost.
\end{definition}

An instance of the ordinal fair $k$-center problem is denoted as $\Ical = \ordfairkcenins$.
The ordinal $k$-center problem is defined when no fairness constraints are enforced.
The fair $k$-center problem is defined analogously, with the key distinction that the distance function $d$ of the underlying metric space is fully known.  
For brevity, we denote the ordinal $k$-center problem with and without fairness constraints by \ordfairkcen\ and \ordkcen, respectively.

We evaluate the quality of the solution via distortion~\cite{procaccia2006distortion}, defined as the approximation ratio between the social cost of the algorithm's solution 
and the social cost of the optimal solution (computed with access to all pairwise distances).
We first formally define distortion for \ordfairkcen, adopting the notion from Burkhardt et~al.~\cite{burkhardt2024low}.

\begin{definition}[Distortion of Algorithm $\alg$]
    For a metric space $(U,d) \in \Mcal$ and an ordinal profile $\succ_U \in \profile(d)$, let $\Ical_{\succ_U}$ be the collection of instances of the ordinal (fair) $k$-center problem defined on points $U$ and ordinal profile $\succ_U$. 
    Let $\alg(I \mid d)$ denote the solution returned by algorithm $\alg$ on instance $I \in \Ical_{\succ_U}$ with underlying metric $d$. 
    Let $\OptSolSet_{I, d}$ be an optimal solution to $I$ when the underlying metric is $d$.
    The distortion of $\alg$ is defined as:
    \[
    \distortion(\alg)\coloneqq \sup_{\substack{(U,d) \in \Mcal \\ \succ_U \in \profile(d)}} \sup_{I \in \Ical_{d,\succ_U}}  \frac{\cost(\alg(I \mid d))}{\cost(\OptSolSet_{I,d})}
    \]
\end{definition}

Since without access to cardinal information, it is not possible to obtain an algorithm with bounded distortion for \ordkcen~\cite{burkhardt2024low}, we allow the algorithm to make \emph{query access} to $d$: given a pair $u,v\in U$, the algorithm can query the distance $d(u,v)$. 
Hence, our goal is to design an algorithm $\alg$ that uses a limited number of queries while ensuring that the solution quality is provably close to the optimal solution.

\xhdr{Reduction} 
For simplicity of exposition, we transform an instance $\Ical = (U, k, \succ_U, \GroupSet, \alphavec)$ of the ordinal fair $k$-center problem, where $\GroupSet = \{\Group_1, \dots, \Group_t\}$ and $\alphavec = \{\alpha_1, \dots, \alpha_t\}$, into an equivalent instance $\Ical' = (U', k, \succ_{U'}, \GroupSet', \alphavec')$ with exactly $k$ groups.
Formally, we construct the transformed instance $\Ical'$ as follows. 
Let $r = \sum_{i \in [t]} \alpha_i$. 
For each group $G_i \in \GroupSet$, we make $\alpha_i$ disjoint copies of $G_i$ in $\GroupSet'$ by duplicating each element of $G_i$ precisely $\alpha_i$ times. Next, if $r < k$, then we create $(k-r)$ new groups in $\GroupSet'$ each containing a distinct copy of $U$. We set $\alphavec' = \{1, 1, \dots, 1\}$, so that the fairness constraint requires selecting exactly one center from each group $\Group'_j$ for $j \in [k]$. 
We define the linear orderings in $\succ_{U'}$ as follows. 
For each point $v' \in U'$, let $v \in U$ be the original point of $v'$ (where $v'$ may be $v$ itself or a duplicate copy of $v$). 
We define the linear order $\succ_{v'}$ on $U'$ as an order extension of $\succ_{v}$ on $U$: for any points $a', b' \in U'$ and their corresponding original points $a, b \in U$, we have $a' \succ_{v'} b'$ if and only if $a \succ_{v} b$.
When $a = b$, the ordering of $a'$ and $b'$ is arbitrary.

This transformation is standard in the fair clustering literature and has been used in several prior works~\cite{thejaswi2022clustering, thejaswi2024diversity, gadekar2025fair}, in this transformed instance $\Ical'$, three key properties are relevant for our analysis to hold: ($i$) the cost of optimal solution in $\Ical'$ remains identical to that of the original instance $\Ical$, since the additional points are duplicates; {($ii$) no extra distance queries are required, as duplicate points share identical distances;} and ($iii$) if the algorithm's output for $\Ical'$ contains multiple copies of the same center, we retain a single copy and supplement it with arbitrary 
centers
from the corresponding group to satisfy the fairness constraints, without affecting the theoretical guarantees on the distortion factor.

From now on, we present our algorithmic results for the instance $\Ical'$, and we directly write the instance as $\Ical$ when there is no ambiguity.

\xhdr{Preliminaries}
In what follows, we define the terminology and definitions necessary to present our algorithmic results.
Given a  set $U$ and two subsets $A,B \subseteq U$ of points, we say that $A$ hits $B$ if $A \cap B \neq \emptyset$. More generally, for an integer $i \geq 0$, $A$ hits $B$ $i$ times if $|A \cap B|=i$. For an ordered set $T = (t_1,\dots,t_k) \subseteq U$, and $\ell \in [k]$, we denote by $T_\ell$ the $\ell$-length prefix of $T$, \ie, $T_\ell= (t_1,\dots, t_\ell)$.

To characterize the quality of Gonzalez-type greedy algorithms~\cite{gonzalez1985clustering, burkhardt2024low}, we apply the notion of progressive cover and critical index.
Let $\mathbf{\Pi}^*$ denote the partition of $U$ induced by the optimal solution $S^*$; \ie, each cluster contains all points whose nearest center in $S^*$ is the same.
Note that the $k$ centers $T$ obtained by Gonzalez's classic $k$-center algorithm~\cite{gonzalez1985clustering} form a progressive $2$-cover. 
In this case, the critical index for $T$ is the largest $\ell \in [k]$ such that $T_{\ell}$ hits each part of the partition $\mathbf{\Pi}^*$ at most once, \ie, the largest $\ell$ that satisfies (i) also satisfies (ii) in the later definition. 
This property that the largest $\ell$ that satisfies (i) also satisfies (ii) also holds for the solution obtained by the algorithm of Burkhardt et~al.~\cite{burkhardt2024low} for the ordinal $k$-center problem.

\begin{definition}[Progressive cover and critical index]
Let $\Jcal$ be an instance of \kcenter (or \fairkcenter) and let $S^*$ be an optimal solution with corresponding partition $\mathbf{\Pi}^*$ of $U$. Fix some $\gamma >0$, and consider an ordered set $T=(t_1,\dots,t_k) \subseteq U$ of size $k$. We say $T$ is a \emph{progressive $\gamma$-cover} for $\Jcal$ \wrt $S^*$, if there exists $\ell \in [k]$ such that ($i$) $T_{\ell}$ hits each part of $\mathbf{\Pi}^*$ at most once, and ($ii$) $\cost(T_{\ell}) \leq \gamma\cdot \cost(S^*)$. Furthermore, we say $\ell \in [k]$ as the \emph{critical index} of a progressive $\gamma$-cover $T$ for $\Jcal$  \wrt $S^*$ if $\ell$ is the maximum index such that ($i$) and ($ii$) hold.
\end{definition}

When $S^*$ is implicit from context, we omit saying \wrt $S^*$.
A crucial tool that is used to obtain a feasible (fair) solution from a progressive $\gamma$-cover $T$~\cite{gadekar2025fair,jones2020fair} is to project $T$ onto a fair solution using a matching on a carefully constructed bipartite graph. 
In our paper, we apply the notion of left-perfect matching to define the projection graph. We recall the definition of left-perfect matching in Definition~\ref{definition:left-perfect-matching} in the appendix.
We formally define the projection graph below.
\begin{definition}[Projection graph]
Consider an instance $\Jcal=\fairkcenins$ of \fairkcenter, and let $S \subseteq U, |S|\le k$. Furthermore, let $\ell \in [|S|]$, and $\lambda \in \mathbb{R}_+$. 
We define the \emph{$(\ell, \lambda)$-projection graph}  $H^{\ell}_{\lambda} = (\InitSolSet_{\ell} \cup \MultiGroupSet, E_{\lambda})$, where there is an edge between  $s \in \InitSolSet_{\ell}$ and $\MultiGroup_i \in \MultiGroupSet$ if $d(s, \MultiGroup_i) \le \lambda$. Furthermore, for a fixed $\ell \in [k]$, we define $\lambda_{\ell}$ as the minimum $\lambda$ such that there exists a \emph{left-perfect matching} on $H^{\ell}_{\lambda}$. 
\end{definition}

By using the projection graph, we can obtain a $\gamma+1$-distortion feasible (fair) solution from a progressive $\gamma$-cover $T$~\cite{gadekar2025fair,jones2020fair}.

\section{$3$-distortion with $\bigO(k^2)$ queries}
\label{sec:conventional-algorithms}

As a warm-up, we first revisit the algorithm for the fair $k$-center problem (\fairkcenter)~\cite{gadekar2025fair} and, by combining it with the techniques of \citet{burkhardt2024low} for the ordinal $k$-center (\ordkcen) setting, we present a 3-distortion algorithm for \ordfairkcen using $O(k^2)$ distance queries.

\xhdr{\large The fair $k$-center algorithm}
First, the algorithm computes an ordered sequence \InitSolSet of $k$ centers using Gonzalez's $k$-center algorithm~\cite{gonzalez1985clustering}
by selecting an arbitrary point as the first center and then choose the point that is farthest from all previously selected centers for $k-1$ iterations.
Then, for each $\ell \in [k]$, it iteratively constructs a \emph{$(\ell, \lambda)$-projection graph} $H^{\ell}_{\lambda}$ over $\lambda$ ranging across all candidate distances from the input.
By doing so, it finds the minimum $\lambda$, denoted $\lambda_{\ell}$, such that there exists a left-perfect matching in $H^{\ell}_{\lambda_{\ell}}$.
We have the following guarantee due to~\citet{gadekar2025fair}.

\begin{theorem}[Theorem 3.1~\cite{gadekar2025fair}]
 \label{corollary:lambda-ell-star-lower-bound}
    Consider $\InitSolSet$, the (ordered) set of $k$ centers obtained using Gonzalez's $k$-center algorithm~\cite{gonzalez1985clustering}.
    Then, $\InitSolSet$ is a progressive $2$-cover for a fair $k$-center instance $\Jcal$. 
    Let $\ell^*$ be the critical index of $\InitSolSet$.
    Furthermore, we have $\lambda_{\ell^*} \leq \cost(\OptSolSet)$ and $\cost(\InitSolSet_{\ell^*}) \leq 2\, \cost(\OptSolSet)$, where $\OptSolSet$ is an optimal solution for $\Jcal$. 
\end{theorem}

Given $\InitSolSet_{\ell}$, we can construct a feasible solution $\SolSet_{\ell}$ as follows:
Given $\ell \in [k]$ and $\lambda_{\ell}$, we can construct the left-perfect matching on the $(\ell, \lambda_{\ell})$-parameterized bipartite graph $H^{\ell}_{\lambda_{\ell}}$.
Suppose that $s \in \InitSolSet_{\ell}$ is matched to $\MultiGroup_i$ in the left-perfect matching, we add the point in $\MultiGroup_i$ that is closest to $s$ into $\SolSet_{\ell}$. 
We then add one arbitrary point from the groups that are not matched to $\InitSolSet_{\ell}$.
This consists a feasible solution 
according to the reduction in Section~\ref{sec:problem-formulation}.
As a direct implication of \Cref{corollary:lambda-ell-star-lower-bound}, there exists some $\tilde{\ell}$ for which $\SolSet_{\tilde{\ell}}$ is a $3$-approximate solution.

\begin{corollary}
    \label{corollary:conventional-lambda-ell-star-lower-bound}
    Let $\tilde{\ell} = \arg\min_{\ell} \big(\cost(\SolSet_{\ell}) + \lambda_{\ell}\big)$.
    The solution $\SolSet_{\tilde{\ell}}$ is a $3$-approximate solution for the fair $k$-center problem.
\end{corollary}

\xhdr{\large A $3$-distortion algorithm with $\bigO(k^2)$ queries}
Next, we summarize the $2$-distortion algorithm of~\citet{burkhardt2024low} for \ordkcen, which is an adaptation of Gonzalez's algorithm to the ordinal setting. At each iteration, when selecting a new center---the point farthest from the current centers---the algorithm queries distances to the farthest points in the ordinal rankings and selects the farthest point based on these queries. We restate this result in Theorem~\ref{theorem:ordinal-k-center-algorithm} below.

\begin{theorem}[Theorem 3.1~\cite{burkhardt2024low}]
\label{theorem:ordinal-k-center-algorithm}
There exists a deterministic $2$-distortion algorithm for $k$-center that makes $\frac{k^2 - k}{2}$
queries.
\end{theorem}

We present a $3$-distortion algorithm by combining the approach of \citet{gadekar2025fair} and \citet{burkhardt2024low} with the pseudocode in Algorithm~\ref{alg:conventional-fair-k-center}.
We present the proof of the algorithm in the appendix.

\begin{algorithm}[t]
\SetAlgoNoEnd
    \caption{A $3$-distortion algorithm}
    \label{alg:conventional-fair-k-center}
    \SetAlgoLined
    \KwIn{An instance of \ordfairkcen, $\Jcal = \ordfairkcenins$.}
    \KwOut{A subset of $k$ items $\SolSet \subseteq U$.}
    Burkhardt et al.'s $2$-distortion $k$-center algorithm: compute a sequence of $k$ centers, \InitSolSet \label{alg:cov:burkhardt}

    Query all the distances of $d(s, \MultiGroup)$ for all $s \in \InitSolSet$ and $\MultiGroup \in \MultiGroupSet$

    \For{$\ell = 1, 2, \dots, k$}{

        \For{$(s, \MultiGroup) \in \InitSolSet_{\ell} \times \MultiGroupSet$}{

        Set $\lambda = d(s, \MultiGroup)$;
        Construct the \emph{$(\ell, \lambda)$-projection graph} $H^{\ell}_{\lambda}$;
        Check whether there exists a left-perfect matching in $H^{\ell}_{\lambda}$; if so, set $\lambda_{\min} = \lambda$.
            
        }
        Set $\lambda_{\ell} = \lambda_{\min}$. Construct the solution $\SolSet_{\ell}$ as described in the main text.
    }

    Return the solution $\SolSet_{\tilde{\ell}}$, where $\tilde{\ell} = \argmin_{\ell} \big(\cost(\SolSet_{\ell}) + \lambda_{\ell}\big)$. 
\end{algorithm}

\begin{restatable}{theorem}{naiveapproach}
\label{lemma:naive-approach}
Algorithm~\ref{alg:conventional-fair-k-center} is a $3$-distortion algorithm for the fair $k$-center problem that takes $2k^2 \in \bigO(k^2)$ distance queries.
\end{restatable}

\section{$5$-Distortion with $\bigO(k\log^2 k)$ queries}
\label{sec:analysis-ordinal-fair-k-center-with-limited-distance-queries}

Naturally, we would like to reduce the number of distance queries while maintaining the same distortion. However, the matching-based approach requires first obtaining a sufficiently good solution to the $k$-center problem and then mapping it to a feasible solution of the fair $k$-center problem. If we stick to this approach, we need to reduce the number of queries required to find a $2$-approximate solution for $k$-center, improving upon the results of~\citet{burkhardt2024low}; which, to our best knowledge, is still an open problem.

This motivates the following \emph{less ambitious}, yet still meaningful, question: \emph{Can we design a deterministic constant distortion algorithm for the ordinal fair $k$-center problem that uses $o(k^2)$ distance queries?}
In this section, we design such an algorithm, answering it in the affirmation. The algorithm we design starts from Burkhardt et al.'s $4$-distortion $k$-center algorithm~\cite{burkhardt2024low}, which uses only $2k$ distance queries. For completeness, we restate the result in Theorem~\ref{theorem:ordinal-k-center-algorithm-smaller-queries}.

\begin{theorem}[Theorem 3.3~\cite{burkhardt2024low}]
    \label{theorem:ordinal-k-center-algorithm-smaller-queries}
    There exists a deterministic $4$-distortion algorithm for $k$-center that makes $2k$ distance queries. 
\end{theorem}

Succinctly, we denote by \bcfrss the algorithm of \citet{burkhardt2024low} corresponding to~\Cref{theorem:ordinal-k-center-algorithm-smaller-queries}.
The main challenge consists of two parts: (i) finding a $\hat{\ell}$ such that mapping the solution $T_{\hat{\ell}}$ to the conventional $k$-center without the fairness constraints to the feasible solution $S$ according to bipartite matching obtains a $5$-approximation result; (ii) regarding the bipartite matching part, we need to adapt their algorithm to the ordinal model that can efficiently construct the bipartite graph $H^{\hat{\ell}}_{\lambda}$ with the smallest possible $\lambda$ that admits a left-perfect matching, given only access to ordinal information and a limited number of distance queries.
We will show that both parts can be achieved with significantly fewer queries than the naive $\bigO(k^2)$ bound.

The pseudocode of our main algorithm is presented in~\Cref{alg:ordinal-fair-k-center-quality-ell}. 
We structure our analysis in two parts. 
In the first part (Section~\ref{sec:4-approximate-k-center-algorithm}--\ref{sec:finalphase}), we present our algorithm, which proceeds in three phases. 
For each phase, we analyze the number of calls to costly subroutines (measured in terms of distance queries) needed and establish the distortion guarantee of the algorithm.
In the second part (Section~\ref{sec:evaluating-the-predicate}--\ref{section:finding-minimum-lambda}), we derive precise bounds on the query complexity of these subroutines, and put everything together to obtain the query complexity of the whole algorithm.
Next, we start with a high-level summary of the analysis.
\subsection{Overview}
\sbpara{\large Part $1$: The algorithm}

\sbpara{Initial phase: } In line~\ref{algo:init}, we compute a set $\InitSolSet$ of $k$ centers using \citet{burkhardt2024low}'s $4$-distortion algorithm (\Cref{theorem:ordinal-k-center-algorithm-smaller-queries}) that takes $2k$ distance queries.

\sbpara{Main Phase: } (lines~\ref{algo:pred}-\ref{algo:mainend})
This is the crucial phase of our algorithm and it starts by defining the predicate $\mathsf{P}(\ell) ~\equiv~ \bigl( 4\,\lambda_{\ell} \leq \cost(\InitSolSet_{\ell}) \bigr)$. 
This phase relies on the fact that, with respect to $\ell$, $\lambda_{\ell}$ is non-decreasing and $\cost(\InitSolSet_{\ell})$ is non-increasing (\Cref{lemma:monotonicity-properties-of-lambda-ell}). 
    Based on the monotonicity properties, we note that the predicate $\mathsf{P}(\ell)$ is monotone in $\ell$, it is initially true and becomes false after some $L$.
    The goal of the algorithm is to output $\hat{\ell}$ such that $\lambda_{\hat{\ell}} \leq \cost(\OptSolSet)$ and $\cost(\InitSolSet_{\hat{\ell}}) \leq 4\, \cost(\OptSolSet)$.
    Considering the corner cases, the mainphase itself is divided into two parts.

\noindent (i) \textit{Corner-cases} part (lines~\ref{algo:base1}-\ref{algo:basek}):
This part corresponds to the corner cases of the predicate  $\mathsf{P}(\ell) ~\equiv~ \bigl( 4\,\lambda_{\ell} \leq \cost(\InitSolSet_{\ell}) \bigr)$ defined in~\cref{algo:pred}. 
If $P(1)$ is false then it implies case (a) of~\Cref{lemma:monotonicity-of-lambda-and-cost-init-sol-set} holds, and hence, $\cost(\InitSolSet_{1}) \leq 4\,\lambda_{1} \leq 4\,\cost(\OptSolSet)$. 
On the other hand, if $P(k)$ is true then it implies case (b) of~\Cref{lemma:monotonicity-of-lambda-and-cost-init-sol-set} holds, and hence $4\,\lambda_{k} \leq \cost(\InitSolSet_{k}) \leq 4\,\cost(\OptSolSet)$.  

\noindent(ii) \textit{Binarysearch part} (lines~\ref{algo:binstart}-\ref{algo:binend}): In this part, the algorithm performs a binary search on $\{2,\dots, k-1\}$ to find the largest $L$ such that $\mathsf{P}({L})$ is true, but $\mathsf{P}({L}+1)$ is false. 
Furthermore, according to (c) of \Cref{lemma:monotonicity-of-lambda-and-cost-init-sol-set}, $L$ satisfies $\min \{\cost(\InitSolSet_{L}), 4 \lambda_{L+1}\} \leq 4\,\cost(\OptSolSet)$. 
Thus by setting $\hat{\ell} = L$ or $\hat{\ell} = L+1$, we have $\lambda_{\hat{\ell}} \leq \cost(\OptSolSet)$ and $\cost(\InitSolSet_{\hat{\ell}}) \leq 4\, \cost(\OptSolSet)$. 
This requires $O(\log k)$ many evaluations of the predicate $\mathsf{P}$, by binary search on $\ell$. 
To compare $\lambda_{L+1}$ and $4\,\cost(\InitSolSet_{L})$, the algorithm calls the subroutine \findlambda\ that computes $\lambda_{L+1}$ \emph{only once} as in Section~\ref{section:finding-minimum-lambda}.

\sbpara{Final Phase: } (lines~\ref{algo:finphst}-\ref{algo:finalphend}) 
In this phase, first the algorithm computes $\lambda_{\hat{\ell}}$ by calling the \findlambda\ subroutine;  the goal is to construct a $(\hat{\ell}, \lambda_{\hat{\ell}})$-projection graph $H^{\hat{\ell}}_{\lambda_{\hat{\ell}}}$ that maps the initial solution $\InitSolSet_{\hat{\ell}}$ to the feasible solution $\SolSet$.
Furthermore, we can bound $\cost(\SolSet) \leq \lambda_{\hat{\ell}} + \cost(\InitSolSet_{\hat{\ell}}) \leq 5\, \cost(\OptSolSet)$, yielding a $5$-distortion result.
We discuss the analysis in detail in in \Cref{theorem:get-feasible-solution}.

\sbpara{\large Part $2$: Query complexity}

\noindent To bound the query complexity, we need to show that given $\ell \in [k]$, how to evaluate $\mathsf{P}(\ell)$, and how to compute $\lambda_\ell$, both efficiently. 
Notice that obtaining the final solution $\SolSet$ does not induce additional queries, as the queries needed to compute $\SolSet$ have already been made when computing $\lambda_{\ell}$; i.e., when constructing the bipartite graph $H^{\ell}_{\lambda_{\ell}}$ and finding the left-perfect matching, the finial solution $\SolSet$ is implicitly obtained from $\InitSolSet_{\ell}$ through the matching.  
First, we show an efficient subroutine to evaluate $\mathsf{P}(\ell)$ for a given $\ell$ with only $\bigO(\ell \log \ell)$ queries.
For computing $\lambda_\ell$, we design a subroutine \findlambda, that uses a \mom\ subroutine, which is based on a binary search approach of finding a \emph{weighted} median of medians.
Note that there are $\ell k$ many values possible for $\lambda_\ell$, corresponding to $\ell$ elements of $\InitSolSet_\ell$, which is expensive to query. However, since we know the linear orderings of of $d(\SolMem, \Group)$ for any fix $\SolMem \in \InitSolSet_\ell$ and all $\Group \in \MultiGroupSet$ in advance, the subroutine can compute the weighted median of the medians of these orderings using only $\ell$ queries; in addition, it takes another $\ell \log k$ queries to reduce the search space. Furthermore, after each such computation, the subroutine reduces the search space by a quarter, resulting in $\bigO(k \log ^2 k)$ queries to find $\lambda_\ell$.
Therefore, the query complexity of~\Cref{alg:finding-lambda-ell} is $\bigO(k\log ^2 k)$, since the \textbf{while} loop of the main phase runs $\bigO(\log k)$ times, and each time it evaluates $\mathsf{P}$ and performs other operations, which require overall $\bigO(k\log^2 k)$ queries.

\subsection{Initial Phase}\label{sec:4-approximate-k-center-algorithm}
The first step of the algorithm is to obtain a set $\InitSolSet$ using~\citet{burkhardt2024low} algorithm, \bcfrss, from~\Cref{theorem:ordinal-k-center-algorithm-smaller-queries}, ignoring the group fairness constraints of $\Ical$.
We start our analysis by introducing the following lemma that captures the properties of $\InitSolSet$. 
For $\ell\in [k]$ and $\lambda \in \mathbb{R}_+$, recall that $H^{\ell}_{\lambda}$ is the $(\ell,\lambda)$-projection graph. Also, when $\ell \in [k]$ is fixed, $\lambda_{\ell}$ is the minimum value for which there exists a {left-perfect matching} on $H^{\ell}_{\lambda}$. 

\begin{restatable}{lemma}{lambdaellstarlowerboundsmallerqueries}
    \label{lemma:lambda-ell-star-lower-bound-smaller-queries}
    The set $\InitSolSet$ returned by \bcfrss on instance $\Ical=\ordfairkcenins$ of \ordfairkcen is
    a progressive $4$-cover for $\Ical$.
    Furthermore, for the critical index $\ell^* \in [k]$ of \OptSolSet, it holds that
    $\lambda_{\ell^*} \leq \cost(\OptSolSet)$, and $\cost(\InitSolSet_{\ell^*}) \leq 4\, \cost(\OptSolSet)$, and hence $\cost(\InitSolSet_{\ell^*}) \leq 4\lambda_{\ell^*} + \cost(\InitSolSet_{\ell^*}) \leq 5\, \cost(\OptSolSet)$.
\end{restatable}

\begin{algorithm}[]
\SetAlgoNoEnd
    \caption{$5$-Distortion Algorithm}
    \label{alg:ordinal-fair-k-center-quality-ell}
    \SetAlgoLined
    \KwIn{An instance of \ordfairkcen $\Ical=\ordfairkcenins$ }
    \KwOut{Solution $\SolSet$ to $\Ical$}

   Compute $k$ centers, \InitSolSet through Burkhardt et al.~\cite{burkhardt2024low}'s $5$-distortion algorithm (\Cref{theorem:ordinal-k-center-algorithm-smaller-queries})\label{algo:init}

   Define predicate 
    $\mathsf{P}(\ell) ~\equiv~ \bigl( 4\,\lambda_{\ell} \leq \cost(\InitSolSet_{\ell}) \bigr)$,
     for $\ell\in[k]$.\label{algo:pred}

    \tcp{Recall that $\lambda_{\ell}$ denotes the minimum $\lambda$ such that there exists a left-perfect matching on $H^{\ell}_{\lambda}$ and $\cost(\InitSolSet_{\ell})=\max_{u\in U} d(u,\InitSolSet_{\ell})$, the predicate checks whether $4\,\lambda_{\ell} \leq \cost(\InitSolSet_{\ell})$.}

    \tcp{Binary search on $\ell$ to find the largest $\ell$ such that $\mathsf{P}(\ell)$ holds.}

    \lIf{$\mathsf{P}(1)$ is false}{$\hat{\ell} \leftarrow 1$}\label{algo:base1}

    \lIf{$\mathsf{P}(k)$ is true}{$\hat{\ell} \leftarrow k$}\label{algo:basek}

    \Else{
        Set $L\leftarrow 1$ and $R\leftarrow k$\;\label{algo:binstart}
        
        \While{$L<R$}{
            $M \leftarrow \bigl\lfloor (L+R+1)/2 \bigr\rfloor$

            \lIf{$\mathsf{P}(M)$}{
                $L\leftarrow M$
            }\lElse{
                $R\leftarrow M-1$
            }
        }\label{algo:binend}
        $\lambda_{L+1} \gets \findlambda(\Ical, L+1, T_{L+1})$\;\label{algo:subph2st} %
        
        \lIf{$\cost(\InitSolSet_{L}) \leq 4 \lambda_{L+1}$}{
            $\hat{\ell}\leftarrow L$.
        }\lElse{
            $\hat{\ell}\leftarrow L+1$\label{algo:mainend}
        }
    }

    $\lambda_{\hat{\ell}}\gets \findlambda(\Ical, \hat{\ell}, T_{\hat{\ell}})$\;\label{algo:finphst}

    $\SolSet \gets $ solution obtained from~\Cref{theorem:get-feasible-solution}\;\label{algo:finalphend}

    Return $\SolSet$
\end{algorithm}

Slightly abusing the terminology, we say that $\lambda_\ell$ and $\cost(\InitSolSet_{\ell})$ are functions of $\ell$. The following lemma says that both $\lambda_{\ell}$ and $\cost(\InitSolSet_{\ell})$ are monoton in $\ell$, which is our key idea in designing binary search on $\ell$ by defining the predicate $\mathsf{P}(\ell) ~\equiv~ \bigl( 4\,\lambda_{\ell} \leq \cost(\InitSolSet_{\ell}) \bigr)$. 

\begin{restatable}{lemma}{monotonicitypropertiesoflambdaell}
    \label{lemma:monotonicity-properties-of-lambda-ell}
    Consider the set $\InitSolSet$ returned by \bcfrss. Then,    $\lambda_{\ell}$ is a non-decreasing function in $\ell$ and $\cost(\InitSolSet_{\ell})$ is a non-increasing function in $\ell$.
\end{restatable}

We put the proofs of \Cref{lemma:lambda-ell-star-lower-bound-smaller-queries} and \Cref{lemma:monotonicity-properties-of-lambda-ell} to the appendix.

\subsection{Main Phase}\label{sec:mainphase}
In the \Cref{sec:4-approximate-k-center-algorithm}, 
as a direct implication of Lemma~\ref{lemma:lambda-ell-star-lower-bound-smaller-queries}, for $\InitSolSet$ obtained from \bcfrss on \ordfairkcen instance $\Ical$, we know that 
there exists an $\ell^*$ which holds that 
$\lambda_{\ell^*} \leq \cost(\OptSolSet)$ and $\cost(\InitSolSet_{\ell^*}) \leq 4\, \cost(\OptSolSet)$, where $S^*$ is an optimal solution for $\Ical$.
However, it is not clear how to find $\ell^*$ as the algorithm does not know the optimal partition $\Pi^*$ in advance. 
As a result, we opt to find an alternative $\hat{\ell}$ that satisfies $\lambda_{\hat{\ell}} \leq \cost(\OptSolSet)$ and $\cost(\InitSolSet_{\hat{\ell}}) \leq 4\, \cost(\OptSolSet)$ as well, this constitutes the main goal of this phase. 

Here we define a predicate $\mathsf{P}(\ell) ~\equiv~$ $( 4\,\lambda_{\ell} \leq \cost(\InitSolSet_{\ell}))$, for $\ell\in[k]$. 
The monotonicity of $\mathsf{P}(\ell)$ facilitates a binary search on $\ell$ to find such $L$ such that $\mathsf{P}(L)$ is true and $\mathsf{P}(L+1)$ is false.
Furthermore, we will show that either $L$ or $L+1$ is the desired $\hat{\ell}$. 
We present our algorithm in \Cref{alg:ordinal-fair-k-center-quality-ell}. 
Let us first present our main result of the phase, \Cref{lemma:binary-search-on-ell}, which states the correctness of the algorithm, and also serves as a framework for the analysis of the query complexity.

\begin{theorem}
    \label{lemma:binary-search-on-ell}
    In Algorithm~\ref{alg:ordinal-fair-k-center-quality-ell}, it holds that $\lambda_{\hat{\ell}} \leq \cost(\OptSolSet)$ and $\cost(\InitSolSet_{\hat{\ell}}) \leq 4\, \cost(\OptSolSet)$.
    In addition, the algorithm needs to evaluate the predicate at most $\lceil \log(k+1) \rceil +2$ times.
\end{theorem}

The monotonicity of $\mathsf{P}(\ell)$ relies on observation that $\lambda_{\ell}$ is a non-decreasing function in $\ell$ and $\cost(\InitSolSet_{\ell})$ is a non-increasing function in $\ell$, as in \Cref{lemma:monotonicity-properties-of-lambda-ell}.
The correctness of the output of the binary search on $\ell$ is established in \Cref{lemma:monotonicity-of-lambda-and-cost-init-sol-set}.
We leave the proof of \Cref{lemma:binary-search-on-ell} to the end of this section, and put the proof \Cref{lemma:monotonicity-of-lambda-and-cost-init-sol-set} in the appendix.

\begin{restatable}{lemma}{monotonicityoflambdacost}
    \label{lemma:monotonicity-of-lambda-and-cost-init-sol-set}
    Let $\InitSolSet$ be the solution of \bcfrss on an instance $\Ical=\ordfairkcenins$ of \ordfairkcen, ignoring the group fairness constraints. Also, let $S^*$ be an optimal solution to $\Ical$. 
    Then, at least one of the three cases  holds:
\begin{itemize}
    \item[\textbf{(a)}] $\cost(\InitSolSet_{1}) \leq 4\,\lambda_{1}$. Moreover, in this case, $\lambda_{1} \leq \cost(\OptSolSet)$. 
    \item[\textbf{(b)}] $4\,\lambda_{k} \leq \cost(\InitSolSet_{k})$. Moreover, in this case, $\cost(\InitSolSet_{k}) \leq 4\,\cost(\OptSolSet)$. 
    \item[\textbf{(c)}] There exists an $2 \leq \ell' \leq k-1$ such that
    \[
        4\lambda_{\ell'} \leq \cost(\InitSolSet_{\ell'})
        \quad \text{and} \quad
        \cost(\InitSolSet_{\ell'+1}) \leq 4\, \lambda_{\ell'+1}.
    \]
    Moreover, in this case, we have
    \[
        \min \left\{ \cost(\InitSolSet_{\ell'}),~ 4\lambda_{\ell'+1} \right\} \leq 4\,\cost(\OptSolSet).
    \]
\end{itemize}
\end{restatable}

Next, we present the illustration of the binary search, which consists of the corner-case part and the binary search on $\ell$ part. 

\subsubsection{Corner-case part}\label{sec:basecase}
In this part, the algorithm checks if either $\mathsf{P}(1)$ is false or $\mathsf{P}(k)$ is true.
For the former case, we have that $\cost(\InitSolSet_{1}) < 4 \lambda_1$, which means $\cost(\InitSolSet_{1}) \leq 4\,\lambda_{1}$. Therefore,  case (a) of~\Cref{lemma:monotonicity-of-lambda-and-cost-init-sol-set}  holds,  hence $\lambda_1 \leq \cost(\OptSolSet)$.
In this case, the algorithm sets $\hat{\ell}=1$.
    When $P(k)$ is true, we have that case (b) of~\Cref{lemma:monotonicity-of-lambda-and-cost-init-sol-set} holds, hence it holds that $4\,\lambda_{k} \leq \cost(\InitSolSet_{k})$ and $\cost(\InitSolSet_{k}) \leq 4\,\cost(\OptSolSet)$.   In this case, the algorithm sets $\hat{\ell}=k$. In both cases, \Cref{lemma:binary-search-on-ell} is true.

\subsubsection{Binary search on $\ell$}
\label{sec:binary-search-on-ell}

To overcome the extended query complexity, we design a binary search approach on $\ell$. 
By the monotonicity of both $\lambda_{\ell}$ and $\cost(\InitSolSet_{\ell})$, the predicate, i.e., $P(\ell)$, is a non-increasing function in $\ell$. 
Namely, there exists such $L$ that $P(L)$ is true and $P(L+1)$ is false. 
Therefore, we can conduct a binary search on $\ell$ to find such $L$. 
Essentially, this case falls into case (c) of~\Cref{lemma:monotonicity-of-lambda-and-cost-init-sol-set}.
The algorithm outputs $\hat{\ell}$ as either $L$ or $L+1$: if $\cost(\InitSolSet_{L}) \leq \lambda_{L+1}$ it outputs $L$ and $L+1$ otherwise.

\smallskip 

Combining the \emph{corner-case part} and the \emph{binary search on $\ell$} part, we obtain a very good guarantee on the solution quality while using significantly fewer distance queries (as established in Section~\ref{sec:evaluating-the-predicate} and Section~\ref{section:finding-minimum-lambda}). 
We now formally prove \Cref{lemma:binary-search-on-ell}.

\textsc{Proof of~\Cref{lemma:binary-search-on-ell}.}
    Algorithm~\ref{alg:ordinal-fair-k-center-quality-ell} outputs $\hat{\ell}$ on three cases. 
    Case 1 that $P(1)$ is false, it implies that $\cost(\InitSolSet_{1}) < 4 \lambda_1$. 
    As a result, it holds true that $\cost(\InitSolSet_{1}) \leq 4\,\lambda_{1}$. 
    Since $\cost(\InitSolSet_{1}) \leq 4\,\lambda_{1}$ holds true, it implies that case (a) holds. 
    Case 2 that $P(k)$ is true directly implies case (b) holds $4\,\lambda_{k} \leq \cost(\InitSolSet_{k})$.  
    In Case 3, the algorithm conducts a binary search on $\ell$ and find the maximum $\ell$ (i.e. the $L$ in the algorithm) such that $P(\ell)$ holds. This implies case (c) holds $4\lambda_{\ell} \leq \cost(\InitSolSet_{\ell})$ and $\cost(\InitSolSet_{\ell + 1}) \leq 4\, \lambda_{\ell +1}$.
    In addition, the algorithm outputs $\hat{\ell}$ as either $L$ or $L+1$: if $\cost(\InitSolSet_{L}) \leq \lambda_{L+1}$ it outputs $L$ and $L+1$ otherwise.
    In all the three cases, the algorithm outputs $\hat{\ell}$ such that $\lambda_{\hat{\ell}} \leq \cost(\OptSolSet)$ and $\cost(\InitSolSet_{\hat{\ell}}) \leq 4\, \cost(\OptSolSet)$ holds.
    In addition, the algorithm needs to evaluate the predicate at most $\lceil \log(k+1) \rceil$ this is because each time in the binary search, the search space is halved.
\hfill\qed

\subsection{Final Phase}\label{sec:finalphase}
Once the algorithm computes $\hat{\ell}$ that satisfies~\Cref{alg:finding-lambda-ell}, it computes $\lambda_{\hat{\ell}}$ in~\Cref{algo:finphst}.
Next in~\Cref{algo:finalphend}, it computes a feasible solution $\SolSet_{\hat{\ell}}$ by mapping the left-perfect matching on $H^{\hat{\ell}}_{\lambda_{\hat{\ell}}}$ from $\InitSolSet_{\hat{\ell}}$ to $\MultiGroupSet$,  and obtain a $5$-approximate solution to the ordinal fair $k$-center problem.
We leave the proof to the appendix in the full version of the paper. 
 
\begin{restatable}{theorem}{getfeasiblesolution}
    \label{theorem:get-feasible-solution} 
    Given $\hat{\ell} \in [k]$ and $\lambda_{\hat{\ell}}$, we can construct the left-perfect matching on the $(\hat{\ell}, \lambda_{\hat{\ell}})$-projection graph $H^{\hat{\ell}}_{\lambda_{\hat{\ell}}}$, to obtain a feasible solution $\SolSet$ as follows:
    suppose that $s \in \InitSolSet_{\hat{\ell}}$ is matched to $\MultiGroup_i$ in the left-perfect matching, 
    add the point in $\MultiGroup_i$ that is closest to $s$ into $\SolSet$. 
    We then add one arbitrary point from the groups that are not matched from $\InitSolSet_{\hat{\ell}}$. 
    The solution $\SolSet$ is a feasible solution for the fair $k$-center problem with distortion at most $5$.
\end{restatable}

\subsection{Evaluating the predicate} \label{sec:evaluating-the-predicate}
Now, we analyze how many distance queries are needed to evaluate the predicate $\mathsf{P}(\ell) ~\equiv~$ $( 4\,\lambda_{\ell} \leq \cost(\InitSolSet_{\ell}))$ for any given $\ell \in [k]$. 
A naive way to evaluate the predicate is to compute both $\lambda_{\ell}$ and $\cost(\InitSolSet_{\ell})$, and compare the values. 
However, computing each $\lambda_{\ell}$ is costly, as we later show in \Cref{section:finding-minimum-lambda} which costs $\bigO(k \log^2 k)$ distance queries and leading to $\bigO(k^3)$ distance queries in total.
An alternative approach is to view the predicate in its equivalent formulation, i.e., $P(\ell) ~\equiv~$ $( \lambda_{\ell} \leq \frac{1}{4} \cdot \cost(\InitSolSet_{\ell}))$.
This formulation essentially checks whether there exists a left-perfect matching on the $(\ell, \frac{1}{4} \cdot \cost(\InitSolSet_{\ell}))$-projection graph $H^{\ell}_{\frac{1}{4} \cdot \cost(\InitSolSet_{\ell})}$.

With this equivalent formulation, we note that evaluating the predicate only costs the distance queries in two parts. 
The first part is to construct the $(\ell, \frac{1}{4} \cdot \cost(\InitSolSet_{\ell}))$-projection graph, which takes at most $\ell \log(k)$ distance queries, we present the result in \Cref{lemma:construct-graph-queries} with more general setting on $\lambda$.  
The second part is to compute the cost of $\InitSolSet_{\ell}$, which takes $\ell$ distance queries to compute. 
Combining the two parts, we get \Cref{lemma:evaluate-predicate-queries}.  
We leave the proof to the appendix.

\begin{restatable}{lemma}{constructgraphqueries}
    \label{lemma:construct-graph-queries}
    Constructing the $(\ell, \lambda)$-projection graph $H^{\ell}_{\lambda}$ takes at most $\ell \log(k)$ distance queries, for any arbitrary $\lambda \geq 0$ and a fixed $\ell$.
    
\end{restatable}

\begin{restatable}{lemma}{evaluatepredicatequeries}
    \label{lemma:evaluate-predicate-queries}
    Given $\ell \in [k]$, it takes at most $\ell \log(k) + \ell$ distance queries to evaluate the predicate $\mathsf{P}(\ell)$.
\end{restatable}

\subsection{Computing $\lambda_\ell$}
\label{section:finding-minimum-lambda}
\begin{figure*}[htbp]
    \centering
    \begin{tabular}{p{0.45\textwidth}p{0.45\textwidth}}
        \includegraphics[width=0.45\textwidth]{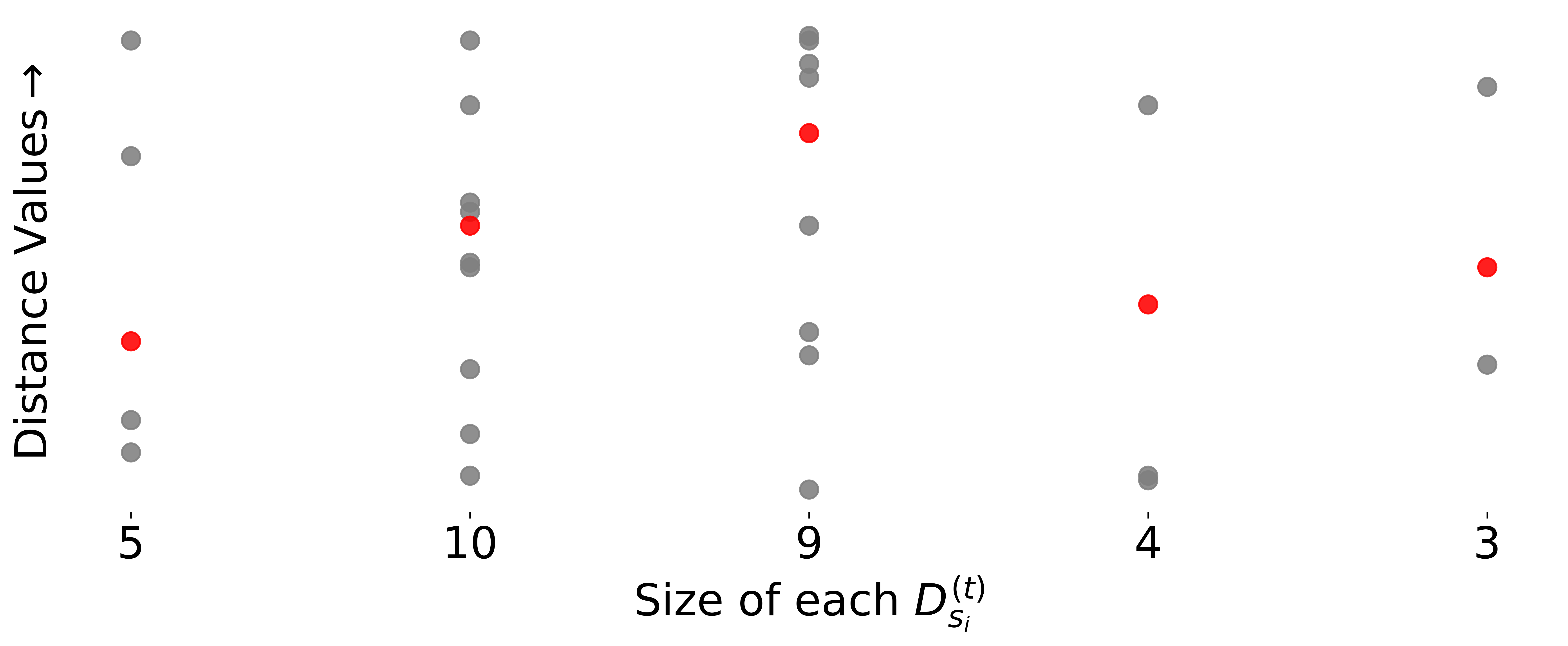} &
        \includegraphics[width=0.45\textwidth]{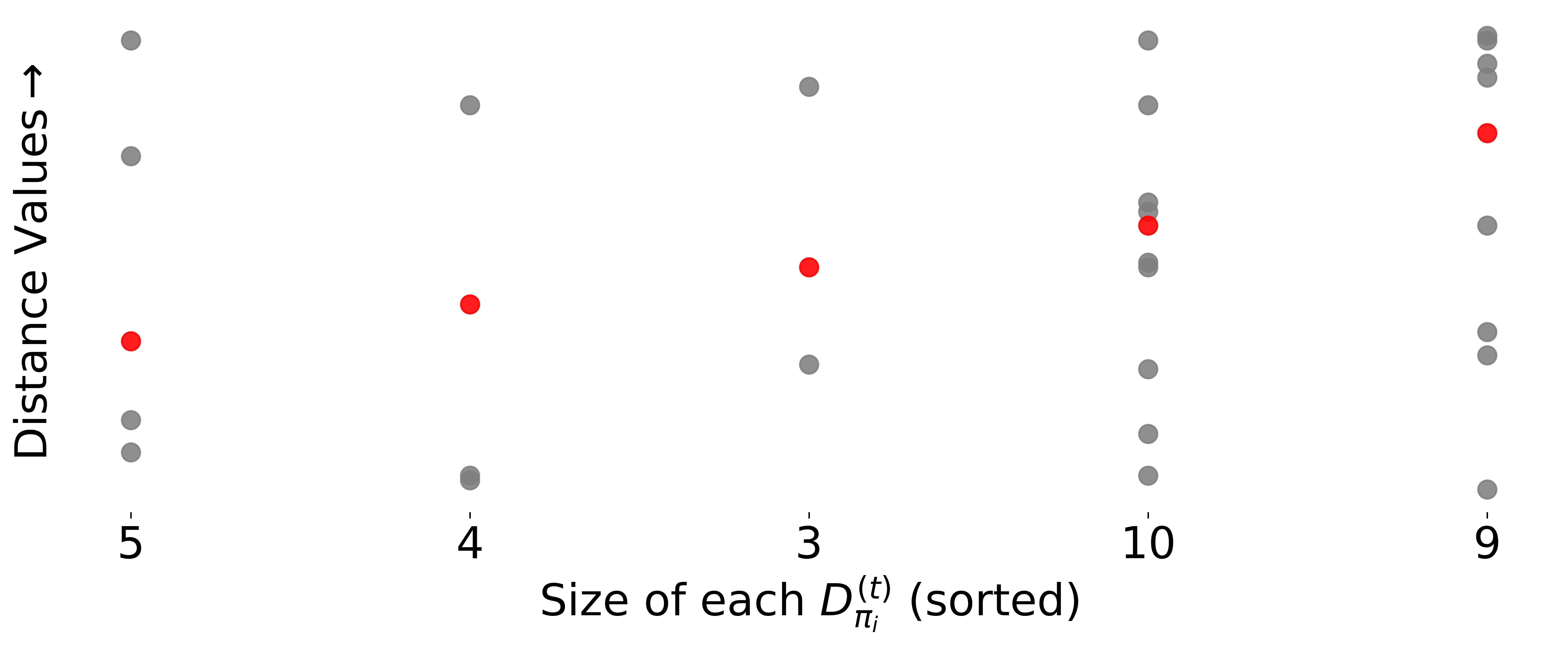} \\
        (a) Each column visualizes the distance values in $D_{s_i}^{(t)}$ for $i \in [\ell]$. The medians $\MED(D_{s_i}^{(t)})$ are marked in red. &
        (b) Sorting the $\{D_{s_i}^{(t)}\}_{i \in [\ell]}$ so that the medians $\MED(D_{s_i}^{(t)})$ are ordered increasingly from left to right. \\
        \includegraphics[width=0.45\textwidth]{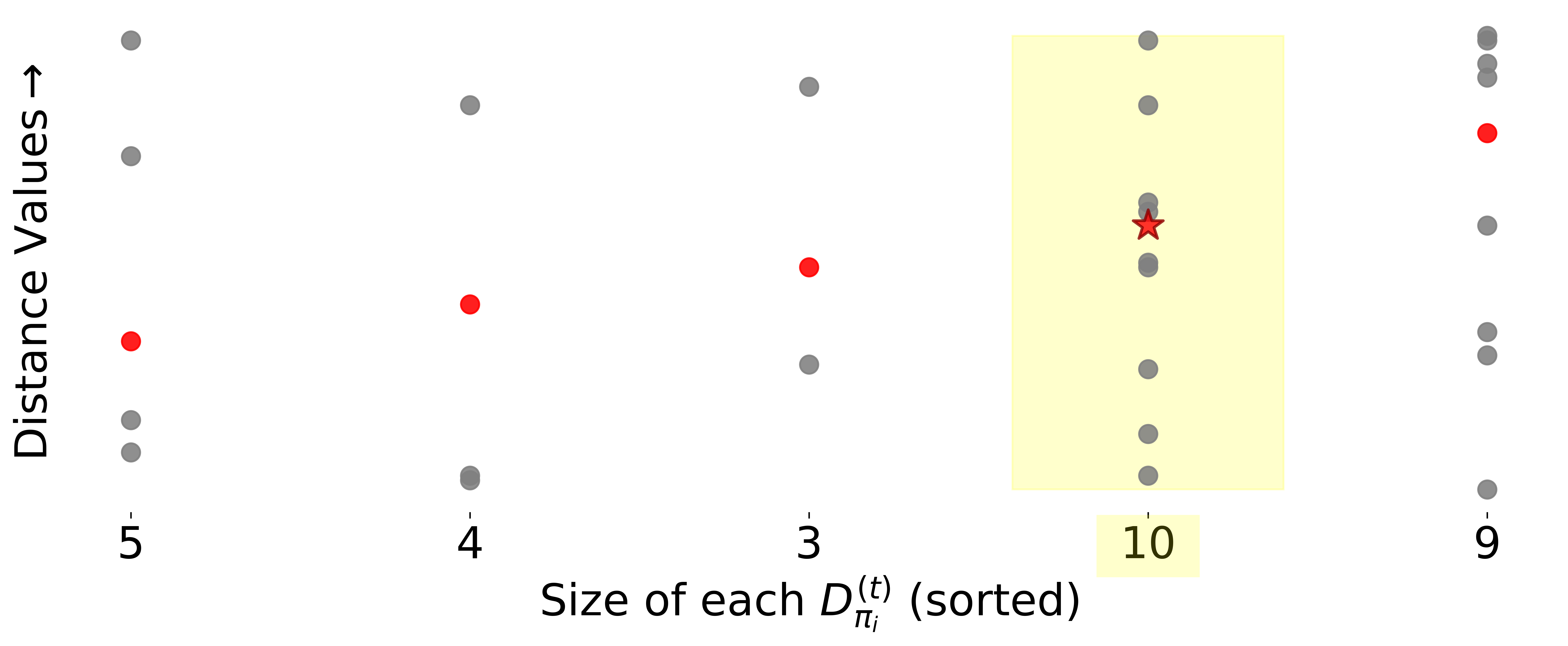} &
        \includegraphics[width=0.45\textwidth]{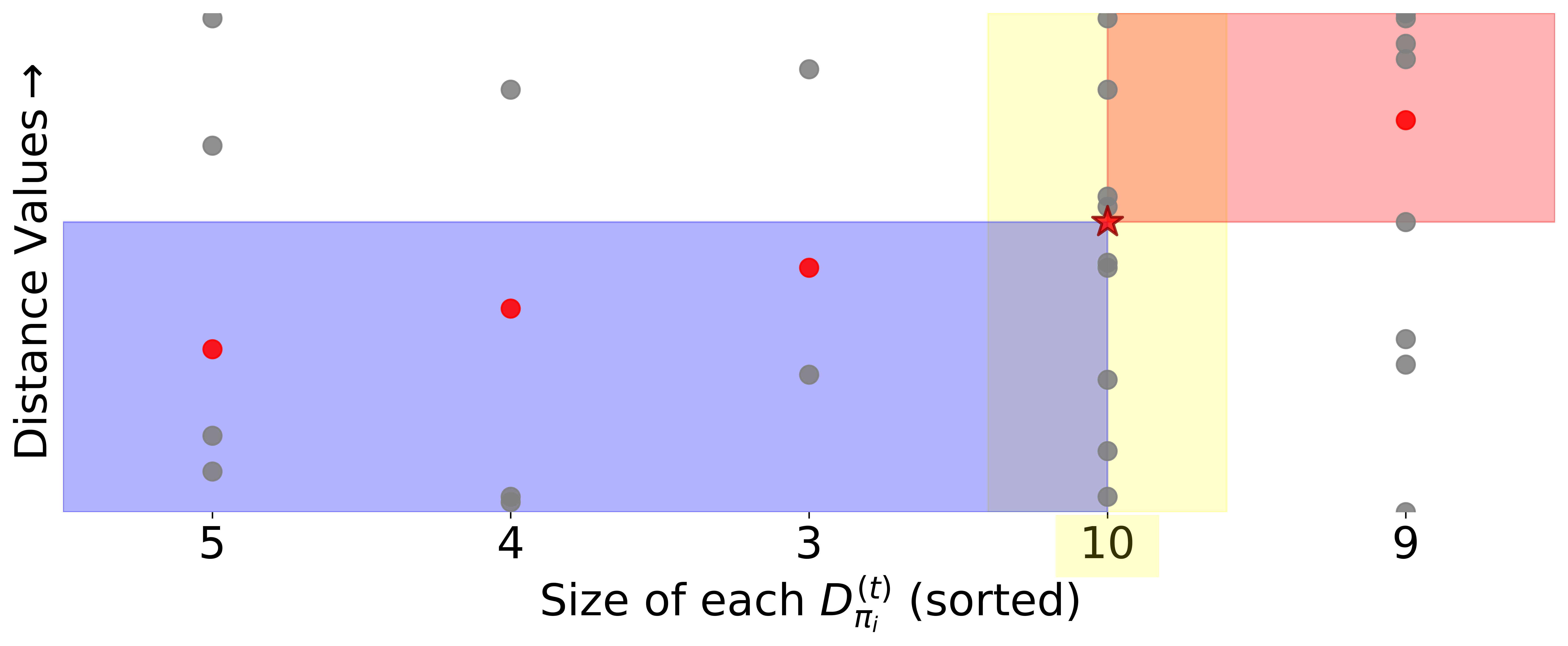} \\
        (c) The \emph{weighted} median corresponding to the sizes of $D_{\pi_i}^{(t)}$ is selected and marked in yellow. The pivot point $\tau$ is the median of the corresponding distances. &
        (d) We observe that least $1/4$ of distance values (in red) are $\geq$ the pivot point $\tau$, and at least $1/4$ (in blue) are $\leq$ the pivot point $\tau$. \\
    \end{tabular}
    \caption{An illustration of the median-of-medians approach.} 
    \label{fig:combined_examples}
\end{figure*}

In \Cref{alg:ordinal-fair-k-center-quality-ell}, the algorithm needs to compute $\lambda_{\ell}$ twice, respectively in line~\ref{algo:subph2st} and line~\ref{algo:finphst}; we need an efficient way to compute $\lambda_{\ell}$ with as few distance queries as possible. 
Towards this, we propose~\Cref{alg:finding-lambda-ell} (\findlambda), that itself is based on \emph{median of medians} approach of~\Cref{alg:median-median} (\mom).
Our main result is that it takes at most $\bigO(k \log^2(k))$ distance queries to find $\lambda_{\ell}$, and we state this result in \Cref{theorem:min-lambda-with-perfect-matching}. 
We first give an overview of the algorithm and then present the proof in \Cref{theorem:min-lambda-with-perfect-matching} at the end of the section.

\begin{theorem}
    \label{theorem:min-lambda-with-perfect-matching}
    For any fixed $\ell$, Algorithm~\ref{alg:finding-lambda-ell} takes $\bigO(k \log^2(k))$ distance queries to find $\lambda_\ell$.
\end{theorem}

A naive approach to find $\lambda_{\ell}$ is to query all the distance values between the initial centers $\InitSolSet$ and all the other points, and conduct the standard binary search~\cite{gadekar2025fair}, which takes $nk$ distance queries.  
In our paper, we reduce the search space $D = \{d(\SolMem_i, \MultiGroup_j) \mid i \in [\ell], j \in [k]\}$, note that entries of $D$ is not known and need to be queried; and we reduce the distance queries by utilizing the structure of $D$.
We give an illustration of the algorithm in \Cref{fig:combined_examples}, and we give a detailed explanation in the following.

\sbpara{Step (a)} 
Denote $\InitSolSet_{\ell} = \{\SolMem_1, \SolMem_2, \dots, \SolMem_{\ell}\}$. 
Let the search space of $D$ at $t$-th step be $D^{(t)}$, which is initialized as $D^{(0)} = D$.  
Denote $D_{\SolMem_i}^{(t)} = \{d(\SolMem_i, \MultiGroup_j) \in D^{(t)} \mid j \in [k]\}$ for each $i \in [\ell]$, which is the set of distance values between the center $\SolMem_i$ and the groups in $\MultiGroupSet$ that remain in the search space at the $t$-th step.
For any fixed $s_i$, the ordering of $D_{\SolMem_i}^{(t)} = \{d(\SolMem_i, \MultiGroup_j)\}_{j \in [k]}$ is known to the algorithm.
Denote the median of $D_{\SolMem_i}^{(t)}$ as $\MED(D_{\SolMem_i}^{(t)})$; we break tie arbitrarily.  
This only costs one distance query to obtain the median value. 

\sbpara{Step (b)}
The algorithm sorts $D_{\SolMem_i}^{(t)}$ according to the values of $\MED(D_{\SolMem_i}^{(t)})$ for any $i \in [\ell]$ in non-decreasing order, which does not cost additional distance queries.  
Let the sorted $\{D_{\SolMem_i}^{(t)}\}_{i \in [\ell]}$ be $\{D_{\pi_i}^{(t)}\}_{i \in [\ell]}$.

\sbpara{Step (c)}
Next, we find a \emph{weighted} median of the $\{\abs{D_{\pi_i}^{(t)}}\}_{i \in [\ell]}$ and denote its index as $\pidx$. 
Particularly, we need to find the $\pidx$ such that the following two conditions hold:
$\sum_{i=\pidx}^{k} \abs{D_{\pi_i}^{(t)}} \geq \frac{1}{2}\sum_{i=1}^k \abs{D_{\pi_i}^{{t}}} = \frac{1}{2} \abs{D^{(t)}}$ and $\sum_{i=1}^{\pidx} \abs{D_{\pi_i}^{(t)}} \geq \frac{1}{2}\sum_{i=1}^k \abs{D_{\pi_i}^{{t}}} = \frac{1}{2} \abs{D^{(t)}}$.
This $\pidx$ is picked according to \emph{line 3} of Algorithm~\ref{alg:median-median}.
We denote the pivot point (a weighted median-of-median) as $\tau = \MED(D_{\pi_{\pidx}}^{(t)})$. 

\sbpara{Step (d)}
The $\tau$ selected in Step (c) has a nice property that at least $1/4$ of the distance values are greater or equal to $\tau$, and at least $1/4$ of the distance values are smaller or equal to $\tau$. In this way, it helps us to reduce our search space by the factor at least $\frac{1}{4}$,
the concrete proof we present in \Cref{lemma:median-median}. 

\smallskip 

Concretely, we introduce the {binary search} based on \emph{median of medians} approach, which we present in Algorithm~\ref{alg:median-median}.
We denote $\mom(D^{(t)}, \tau^{(t)})$ as the \emph{median-of-median search} of the search space $D^{(t)}$ at the $t$-th iteration and $\tau^{(t)}$ as the \emph{current best feasible result}. 
Here, the \emph{median-of-median} is not the exact median; however, it is able to reduce the search space by a constant factor, which we show in step (c) and more formally in Lemma~\ref{lemma:median-median}.
The $\mom(D^{(t)}, \tau^{(t)})$ essentially does two things. First, it finds a pivot point $\tau$ as in line~\ref{algo:find-pivot}. Second, it constructs the bipartite graph $H^{\ell}_{\tau}$ and checks if it has a left-perfect matching: if yes, it means that $\lambda_{\ell} \leq \tau$, and $\tau$ is the current best feasible result; if no, it means that $\lambda_{\ell} > \tau$, and $\tau$ is not a feasible result, and the algorithm returns $\tau^{(t)}$ as the current best feasible result.

Next, the \Cref{lemma:reduce-search-space} shows the number of distance queries needed to construct the next search space $D^{(t+1)}$, respectively from line~\ref{algo:reduce-search-space-up} and line~\ref{algo:reduce-search-space-down} that correspond to $\lambda_{\ell} \leq \tau$ and $\lambda_{\ell} > \tau$.
Thus, the next search space $D^{(t+1)}$ consists of either all the distances that are \emph{strictly smaller than $\tau$} (since we already recorded $\tau$ as the current best feasible result) or all the distances that are \emph{strictly greater than $\tau$}. 
Using the same argument as \Cref{lemma:construct-graph-queries}, we can show for a fixed threshold $\tau$, it takes at most $\ell \log(k)$ distance queries to construct the next search space $D^{(t+1)}$, we state it in \Cref{lemma:reduce-search-space}. 
\begin{restatable}{lemma}{reducesearchspace}
    \label{lemma:reduce-search-space}
    For any fixed threshold $\tau$ and the current search space $D^{(t)}$, it takes at most $\ell \log(k)$ distance queries to find the distances that are strictly smaller than $\tau$ or the distances that are strictly greater than $\tau$.
\end{restatable}

\begin{algorithm}[]
    \SetAlgoNoEnd
    \caption{\findlambda}
    \label{alg:finding-lambda-ell}
    \SetAlgoLined
    \KwIn{An instance $\Ical$ of \ordfairkcen, $\ell \in [k]$, and $\InitSolSet_{\ell}$}
    \KwOut{$\lambda_{\ell}$}
    \tcp{Finding $\lambda_{\ell}$ using ordinal information}
    
    Let the search space be defined as a set of distance values of all the possible pairs between the centers in $\InitSolSet_{\ell}$ and the groups in $\MultiGroupSet$, i.e., $D^{(0)} = \{d(\SolMem_i, \MultiGroup_j) \mid i \in [\ell], j \in [k]\}$

    Let $\tau^{0} = \max D $ \tcp{Maximum distance value}

    \For{$t = 0, 1, \dots, \log_{\frac{4}{3}}(k \ell)$}{
    $D^{(t+1)}, \tau^{(t+1)} = \mom(D^{(t)}, \tau^{(t)})$

    \If{$D^{(t+1)}$ is empty}{
        Return $\tau^{(t+1)}$
    }
    }

\end{algorithm}

\begin{algorithm}[h]
\SetAlgoNoEnd
    \caption{The Median of Medians $\mom(D^{(t)})$}
    \label{alg:median-median}
    \SetAlgoLined
    \KwIn{$D^{(t)}$, $\tau^{(t)}$}
    \KwOut{$D^{(t+1)}$, $\tau^{(t+1)}$}

    Sort $\{D_{s_i}^{(t)}\}_{i \in [\ell]}$ in non-decreasing order according to the distance value of $\MED(D_{s_i}^{(t)})$ across $i \in [\ell]$. Here, $\MED(D_{s_i}^{(t)})$ is the median of $D_{s_i}^{(t)}$

    Denote the sorted $\{D_{\SolMem_i}^{(t)}\}_{i \in [\ell]}$ as $\{D_{\pi_i}^{(t)}\}_{i \in [\ell]}$.
    Let the size of $D^{(t)}$ be $L^{(t)}$

    \tcp{Start with $j=1$, and increment $j$ until the condition is met}

    Let $\pidx$ be the smallest $j$ such that $\sum_{i=1}^{j} \abs{D_{\pi_{i}}^{(t)}} \geq \frac{L^{(t)}}{2}$

    Let $\tau = \MED(D_{\pi_{\pidx}}^{(t)})$ and construct the bipartite graph $H^{\ell}_{\tau}$ \label{algo:find-pivot}

    \uIf{$H^{\ell}_{\tau}$ has a left-perfect matching}{
        $\tau^{(t+1)} = \tau$ \tcp{$\tau$ is current best feasible result}
        
        $D^{(t+1)} = \{d \in D^{(t)} \mid d < \tau \}$  \label{algo:reduce-search-space-up}
    }
    \Else{
        $\tau^{(t+1)} = \tau^{(t)}$ \tcp{$\tau^{(t)}$ is best feasible result from the previous step}
        $D^{(t+1)} = \{d \in D^{(t)} \mid d > \tau \}$ \label{algo:reduce-search-space-down}
    }
    Return $D^{(t+1)}$, $\tau^{(t+1)}$
\end{algorithm}

We state our observation from step (d) formally here that at each iteration, the search space is reduced by a constant factor. 

\begin{restatable}{lemma}{medianmedian}
    \label{lemma:median-median}
    Each invocation of Algorithm~\ref{alg:median-median} reduces the search space by at least $\frac{1}{4}$, namely, $\abs{D^{(t+1)}} \leq \frac{3}{4} \abs{D^{(t)}}$.
\end{restatable}

As a direct implication of Lemma~\ref{lemma:median-median}, we get,
\begin{corollary}
    It takes at most $\log_{\frac{4}{3}}(k \ell)$ iterations to find the minimum value of $\lambda_{\ell}$.
\end{corollary}

Now, we are ready to give a correctness proof for Algorithm~\ref{alg:finding-lambda-ell} that it indeed outputs $\lambda_{\ell}$, we state it in \Cref{theorem:correctness-of-finding-lambda-ell}. 
Our proof is by contradiction, and is an extension of the standard binary search proof. 

\begin{theorem}
    \label{theorem:correctness-of-finding-lambda-ell}
    Given $\ell \in [k]$, Algorithm~\ref{alg:finding-lambda-ell} outputs $\lambda_{\ell}$. 
\end{theorem}
\begin{proof}
    We prove by contradiction. 
    Let's make case distinctions. 
    Suppose the algorithm outputs $\lambda' > \lambda_{\ell}$; it indicates at one step, $\lambda_{\ell}$ is removed from the search space but $\lambda'$ is not. 
    Namely algorithm~\ref{alg:median-median} set a threshold $\tau$ such that $\tau \geq \lambda_{\ell}$, and $H^{\ell}_{\tau}$ does not have a left-perfect matching. 
    This is already a contradiction with the definition of $\lambda_{\ell}$. 

    Next, let us suppose the algorithm outputs $\lambda' < \lambda_{\ell}$, again it indicates that at some step $\lambda_{\ell}$ is removed from the search space but $\lambda'$ is not. 
    There are only two scenarios, the first is that at one step $t$, the algorithm sets $\tau = \lambda_{\ell}$ and finds the left-perfect matching. In this scenario, it holds that $\tau^{t'} = \lambda_{\ell}$ for any following step $t'$ after $t$. Contrary to that the algorithm outputs $\lambda'$. 
    The second is that the algorithm sets $\tau < \lambda_{\ell}$ and finds the left-perfect matching, which is not possible by the definition of $\lambda_{\ell}$. 

    Thus we conclude Algorithm~\ref{alg:finding-lambda-ell} outputs $\lambda_{\ell}$. 
\end{proof}

We are now ready to prove Theorem~\ref{theorem:min-lambda-with-perfect-matching}.

\begin{proof}[Proof of Theorem~\ref{theorem:min-lambda-with-perfect-matching}]
    The correctness of the algorithm is proved by Theorem~\ref{theorem:correctness-of-finding-lambda-ell}.
    We focus on the number of distance queries by Algorithm~\ref{alg:finding-lambda-ell}.

    At each iteration, the algorithm conducts the median median search on the search space $D^{(t)}$, reduces the size of the search space to its $\frac{3}{4}$. 
    It takes $k$ to find the maximum distance value in $D$.
    It takes $\ell$ distance queries to sort $\{D_{s_i}^{(t)}\}_{i \in [\ell]}$ in non-decreasing order according to the distance value of $\MED(D_{s_i}^{(t)})$ across $i \in [\ell]$.
    It also takes $2 \ell \log(k)$ distance queries to construct the bipartite graph and to obtain $D^{(t+1)}$. 
    The procedure repeats at most $\log_{\frac{4}{3}} (k \ell)$ times, thus in total $\log_{\frac{4}{3}}(k \ell) (2 \ell \log(k) + \ell) + k \in \bigO(\ell \log^2 k)$ distance queries to find $\lambda_{\ell}$.
\end{proof}
\subsubsection*{Putting everything together we present a $5$-distortion algorithm using $O(k \log^2 k)$ queries.}
\label{sec:putting-everything-together}

\begin{theorem}
    \label{theorem:putting-everything-together} 
    Algorithm~\ref{alg:ordinal-fair-k-center-quality-ell} finds a $5$-distortion solution to the \emph{ordinal fair $k$-center problem} using at most $\bigO(k \log^2 k)$ distance queries.
\end{theorem}
\begin{proof}
    First, in the initial phase, finding an initial solution $\InitSolSet$ takes $2k$ distance queries by \Cref{theorem:ordinal-k-center-algorithm-smaller-queries}.
    Next, in the main phase, the predicate $P$ is called at most $\lceil \log(k + 1) \rceil + 2$ times  Theorem~\ref{lemma:binary-search-on-ell}. 
    To evaluate the predicate, it takes at most $\ell \log(k) + \ell$ distance queries by Lemma~\ref{lemma:evaluate-predicate-queries}. 
    Thus evaluating the predicate takes at most $(\lceil \log(k+1) \rceil + 2) (k \log(k) + k) \in \bigO(k \log^2 k)$ distance queries. 
    In the same main phase, the algorithm needs to compute a concrete $\lambda_{L+1}$ at most once, which takes at most $\log_{\frac{4}{3}}(k (L+1)) (2 (L+1) \log(k) + (L+1)) + k \in \bigO(k \log^2 k)$ distance queries by Theorem~\ref{theorem:min-lambda-with-perfect-matching}.
    In the final phase, again a $\lambda_{\hat{\ell}}$ is computed,  which costs $\bigO(k \log^2 k)$ distance queries. 
    Notice that obtaining the final solution $\SolSet$ does not induce additional queries, as the queries needed to compute $\SolSet$ have already been made when computing $\lambda_{\hat{\ell}}$. 
    Putting everything together, the algorithm takes $\bigO(k \log^2 k)$ distance queries. 
\end{proof}

\section{Conclusions}

In this work, we studied the ordinal k-committee selection problem under limited cardinal information and fairness constraints, where a minimum number of representatives must be chosen from each demographic group under an egalitarian (min–max) social cost objective. Modeling this as the ordinal fair k-center problem, we developed two query-efficient algorithms: a $5$-distortion algorithm using $\bigO(k \log^2 k)$ queries and a $3$-distortion algorithm using $2k^2$ queries.
To our knowledge, this is the first work to incorporate fairness constraints into the ordinal $k$-center framework. Our results highlight new challenges introduced by fairness in ordinal settings and open promising directions for developing efficient algorithms in fair and query-efficient social choice models.

\xhdr{Acknowledgments} 
Gionis and Tu are supported by 
the ERC Advanced Grant REBOUND (834862),
the Swedish Research Council project ExCLUS (2024-05603),
and the Wallenberg AI, Autonomous Systems and Software Program (WASP) funded by the Knut and Alice Wallenberg Foundation.
Thejaswi acknowledges support from the European Research Council (ERC) under the European Union's Horizon $2020$ research and innovation programme (Grant No. $945719$).

\clearpage

\bibliographystyle{ACM-Reference-Format}
\bibliography{ref}

\clearpage

\appendix

\section{Omitted Definitions  }
\begin{definition}[Left-perfect matching]
    \label{definition:left-perfect-matching}
    A matching $M$ on a bipartite graph is \emph{left-perfect} if for every vertex $v$ in the left set of the bipartite graph, there is an edge in $M$ incident to $v$.
\end{definition}

\section{Omitted proofs from Section~\ref{sec:conventional-algorithms}}

\naiveapproach*
\begin{proof}
The distortion of the algorithm follows directly from Corollary~\ref{corollary:lambda-ell-star-lower-bound}. The number of distance queries consists of three parts: First, computing the initial solution $\InitSolSet$ takes $\frac{k^2 - k}{2}$ distance queries by Theorem~\ref{theorem:ordinal-k-center-algorithm}. Second, querying all distances between the centers in $\InitSolSet$ and the groups in $\MultiGroupSet$ takes $k^2$ distance queries. Third, for any fixed $\ell$, the algorithm makes $\ell$ distance queries to compute $\cost(\InitSolSet_{\ell})$, leading to a total of $\frac{k(k+1)}{2}$ queries across all $\ell$. Summing these, the algorithm takes $2 k^2 \in \bigO(k^2)$ distance queries.
\end{proof}

\section{Omitted proofs from Section~\ref{sec:analysis-ordinal-fair-k-center-with-limited-distance-queries}}

\lambdaellstarlowerboundsmallerqueries* 
\begin{proof}
    Our proof utilizes Lemma 3.6 from Burkhardt et al.~\cite{burkhardt2024low}.
    Let us first formally restate this lemma using the notation of our paper. 
    \begin{lemma}[Lemma 3.6~\cite{burkhardt2024low}]
        \label{lemma:init-sol-mem-lower-bound-burkhardt}
        Let $\{\InitSolMem\} = \InitSolSet_{i+1} \setminus \InitSolSet_{i}$ be the center selected at the $i+1$-th step. 
        It holds that $d(\InitSolMem, \InitSolSet_{i}) \geq \frac{1}{2} \max_{u \in U} d(u, \InitSolSet_{i})$. 
    \end{lemma}

We start with proving the $\InitSolSet$ is a progressive $4$-cover for $\Ical$.
Let $\OptSolSet$ be the optimal solution for $\Ical$.
Let $\ell \in [k]$ be the largest index such that $\InitSolSet_{\ell}$ hits each part of $\mathbf{\Pi}^*$ at most once.
If $\ell = k$, then by triangle inequality, for any $u \in U$, and suppose its closest center in $\InitSolSet$ is $\InitSolMem$, and $\InitSolMem$ is optimal cluster with center $\SolMem^*$. It holds that $d(u, \InitSolSet) = d(u, \InitSolMem) \leq d(u, \SolMem^*) + d(\SolMem^*, \InitSolMem) \leq 2\cdot \cost(\OptSolSet)$.
If $\ell < k$, it indicates that $\InitSolMem_{\ell+1}$ is the first center that hits a optimal cluster 
more than once.  
Let $\InitSolMem_{i}$ where $i\leq \ell$ be a previous center that also hits the same optimal cluster. 
It holds that for any $u \in U$, it holds that $d(u, \InitSolSet_{\ell}) \leq 2 d(\InitSolMem_{\ell+1}, \InitSolSet_{\ell}) \leq 2 d(\InitSolMem_{\ell + 1}, \InitSolMem_{i}) \leq 4 \cost(\OptSolSet)$.
The first inequality holds as $\InitSolMem_{\ell+1}$ is a $\frac{1}{2}$-approximate futhest point to $\InitSolSet_{\ell}$, according to Lemma~\ref{lemma:init-sol-mem-lower-bound-burkhardt}.
The second inequality holds as $\InitSolMem_{i}$ is one of the centers in $\InitSolSet_{\ell}$. 
The third inequality holds by triangle inequality and as both $\InitSolMem_{\ell+1}$ and $\InitSolMem_{i}$ are in the same optimal cluster. 
Therefore, $\InitSolSet$ is a progressive $4$-cover for $\Ical$. 
The critical index $\ell^*$ is the $\ell$ we define above, i.e., the largest index such that $\InitSolSet_{\ell}$ hits each part of $\mathbf{\Pi}^*$ at most once.

Next, we prove that $\lambda_{\ell^*} \leq \cost(\OptSolSet)$. 
Since $\InitSolSet_{\ell}$ hits each part of $\mathbf{\Pi}^*$ at most once, it follows that each $\InitSolMem \in \InitSolSet_{\ell}$ can be uniquely mapped to each group $\MultiGroup \in \MultiGroupSet$; it implies $d(\InitSolMem, \MultiGroup) \leq \cost(\OptSolSet)$.
Let us recall the definition of $(\ell,\lambda)$-projection graph $H^{\ell}_{\lambda} = (\InitSolSet_{\ell} \cup \MultiGroupSet, E_{\lambda})$.
If there exists a left perfect matching on $H^{\ell^*}_{\lambda}$, it indicates that for any $\InitSolMem \in \InitSolSet_{\ell^*}$, there exists a $\MultiGroup \in \MultiGroupSet$ such that $d(\InitSolMem, \MultiGroup) \leq \lambda$.
This implies that when $\lambda = \cost(\OptSolSet)$, there exists a left perfect matching on $H^{\ell^*}_{\lambda}$.
Also because that $\lambda_{\ell^*}$ is the minimum possible value to guarantee the existence of a left perfect matching on $H^{\ell^*}_{\lambda_{\ell^*}}$, it follows that $\lambda_{\ell^*} \leq \cost(\OptSolSet)$.
\end{proof}

\monotonicitypropertiesoflambdaell* 
\begin{proof}
    We prove that $\lambda_{\ell}$ is a non-decreasing function in $\ell$ and $\cost(\InitSolSet_{\ell})$ is a non-increasing function in $\ell$ by the way they are defined. 
    Let $\ell < \ell'$.

    Recall $\cost(\InitSolSet_{\ell}) = \max_{u \in U} d(u, \InitSolSet_{\ell}) \geq \max_{u \in U} d(u, \InitSolSet_{\ell'}) = \cost(\InitSolSet_{\ell'})$. 
    The inequality $\geq$ holds as $\InitSolSet_{\ell} \subset \InitSolSet_{\ell'}$, therefore, $d(u, \InitSolSet_{\ell}) \geq d(u, \InitSolSet_{\ell'})$. 
    This implies that $\cost(\InitSolSet_{\ell})$ is a non-increasing function in $\ell$.

    Recall that $H^{\ell}_{\lambda} = (\InitSolSet_{\ell} \cup \MultiGroupSet, E_{\lambda})$, where $E_{\lambda} = \{(\InitSolMem, \MultiGroup) \mid d(\InitSolMem, \MultiGroup) \leq \lambda\}$. 
    $\lambda_{\ell}$ is the minimum $\lambda$ such that there exists a left perfect matching on $H^{\ell}_{\lambda}$.
    Since $\ell < \ell'$, it follows that $\InitSolSet_{\ell} \subset \InitSolSet_{\ell'}$, therefore, whenever there exists a left perfect matching on $H^{\ell'}_{\lambda_{\ell'}}$, there exists a left perfect matching on $H^{\ell}_{\lambda_{\ell'}}$.
    It implies that $\lambda_{\ell} \leq \lambda_{\ell'}$.
\end{proof}

\monotonicityoflambdacost*
\begin{proof}[Proof of Lemma~\ref{lemma:monotonicity-of-lambda-and-cost-init-sol-set}]
    Since by Lemma~\ref{lemma:monotonicity-properties-of-lambda-ell}, $\lambda_{\ell}$ is a non-de\-creas\-ing function and $\cost(\InitSolSet_{\ell})$ is a non-increasing function, w.r.t. $\ell$, we have three cases.
    
    Case (a): When $4 \lambda_{\ell} \geq \cost(\InitSolSet_{\ell})$ for all $\ell \in [k]$, it follows that $4 \lambda_{1} \geq \cost(\InitSolSet_{1})$.
    Furthermore, we have $\lambda_1 \le \lambda_{\ell^*} \le  \cost(\OptSolSet)$, due to~\Cref{lemma:lambda-ell-star-lower-bound-smaller-queries}.

    Case (b): When $4 \lambda_{\ell} \leq \cost(\InitSolSet_{\ell})$ for all $\ell \in [k]$, it follows that $4 \lambda_{k} \leq \cost(\InitSolSet_{k})$. 
    Furthermore, we have $\cost(\InitSolSet_{k}) \leq \cost(\InitSolSet_{\ell^*}) \leq 4\,\cost(\OptSolSet)$,  due to~\Cref{lemma:lambda-ell-star-lower-bound-smaller-queries}.

    Case (c): If the above two cases are false then, there exists $2 \leq \ell' \leq k-1$ such that for any $\ell \leq \ell'$, it holds that $4 \lambda_{\ell} \leq \cost(\InitSolSet_{\ell})$ and for any $\ell \geq \ell' + 1$, it holds that $4 \lambda_{\ell} \geq \cost(\InitSolSet_{\ell})$.
    Next, we show that $\min \{\cost(\InitSolSet_{\ell'}), 4 \lambda_{\ell' + 1}\} \leq 4\,\cost(\OptSolSet)$. From \Cref{lemma:lambda-ell-star-lower-bound-smaller-queries}, we have that for the critical index $\ell^*$, it holds that $\max\{4\lambda_{\ell^*}, \cost(\InitSolSet_{\ell^*})\}\leq 4\, \cost(\OptSolSet)$. First, consider the case when $\ell^* \le \ell'$, and note that since, $ \cost(\InitSolSet_{\ell^*}) \leq 4\, \cost(\OptSolSet)$, we have that $\cost(\InitSolSet_{\ell'}) \le  \cost(\InitSolSet_{\ell^*})\leq 4\, \cost(\OptSolSet)$, as desired. Now, consider the case when $\ell^* \geq \ell'+1$. Again using~\Cref{lemma:lambda-ell-star-lower-bound-smaller-queries}, we have that $4\lambda_{\ell'+1} \leq 4\lambda_{\ell'+1} \leq 4\, \cost(\OptSolSet)$, as required.
\end{proof}

\getfeasiblesolution*
\begin{proof}
    The solution is a feasible solution because $\abs{\SolSet_{\hat{\ell}}} \leq k$ and $\abs{\SolSet_{\hat{\ell}} \cap \MultiGroup_i} \geq \alpha_i$ for all $i \in [k]$.
    The distortion of the algorithm is proved by triangle inequality and Theorem~\ref{lemma:binary-search-on-ell}. 
    By Theorem~\ref{lemma:binary-search-on-ell}, we know that $\lambda_{\hat{\ell}} \leq \cost(\OptSolSet)$ and $\cost(\InitSolSet_{\hat{\ell}}) \leq 4\, \cost(\OptSolSet)$. 
    It follows that $\cost(\SolSet_{\hat{\ell}}) \leq \lambda_{\hat{\ell}} + \cost(\InitSolSet_{\hat{\ell}}) \leq 5\, \cost(\OptSolSet)$ through the left-perfect matching. 
\end{proof}

\constructgraphqueries*
\begin{proof}
    Recall that to construct the edge set $E_{\lambda}$ of $H^{\ell}_{\lambda}$, we need to determine for each center $\SolMem_{i'} \in \InitSolSet_{\ell}$ and each group $\MultiGroup_{j'} \in \MultiGroupSet$ whether $d(\SolMem_{i'}, \MultiGroup_{j'}) \leq \lambda$.
    For any center $\SolMem_{i'} \in \InitSolSet_{\ell}$, we can leverage the fact that the groups in $\MultiGroupSet$ are ordered by their distances to $\SolMem_{i'}$ in non-decreasing order (denoted by $\succ_{\SolMem_{i'}}$). 
    We perform a binary search to find the rightmost group $\MultiGroup_{j^*}$ in this ordering such that $d(\SolMem_{i'}, \MultiGroup_{j^*}) \leq \lambda$.
    This binary search requires at most $\log(k)$ distance queries.
    Once we find $\MultiGroup_{j^*}$, we can determine all edges incident to $\SolMem_{i'}$: there is an edge between $\SolMem_{i'}$ and $\MultiGroup_{j'}$ if and only if $\MultiGroup_{j'} \preceq_{\SolMem_{i'}} \MultiGroup_{j^*}$.
    This is because all groups $\MultiGroup_{j'}$ with $\MultiGroup_{j'} \preceq_{\SolMem_{i'}} \MultiGroup_{j^*}$ have distance at most $\lambda$ (by the ordering property), while all groups $\MultiGroup_{j'}$ with $\MultiGroup_{j'} \succ_{\SolMem_{i'}} \MultiGroup_{j^*}$ have distance strictly greater than $\lambda$.
    Since we perform this procedure for each of the $\ell$ centers in $\InitSolSet_{\ell}$, the total number of distance queries is at most $\ell \log(k)$.
\end{proof}

\evaluatepredicatequeries*
\textsc{Proof of~\Cref{lemma:evaluate-predicate-queries}.}
   For any fixed $\lambda$, it takes at most $\ell \log(k)$ distance queries to construct the bipartite graph $H^{\ell}_{\lambda}$ according to Lemma~\ref{lemma:construct-graph-queries}.
    In addition, $\cost(\InitSolSet_{\ell})$ takes $\ell$ distance queries to compute.
    Checking whether there exists a left-perfect matching on $H^{\ell}_{\lambda}$ does not take additional distance queries. 
\hfill\qed

\reducesearchspace*
\begin{proof}
    The proof closely follows the reasoning in Lemma~\ref{lemma:construct-graph-queries}, with only a slight variation in the binary search condition.  
\end{proof}

\medianmedian* 
\textsc{Proof.}
First, we make an observation on $\pidx$ that $\sum_{i=\pidx}^{k} \abs{D_{\pi_{i}}^{(t)}} \geq \frac{L^{(t)}}{2}$ also holds. 
This is because that $\pidx$ is the smallest $j$ such that $\sum_{i=1}^{j} \abs{D_{\pi_{i}}^{(t)}} \geq \frac{L^{(t)}}{2}$, it follows that $\sum_{i=1}^{\pidx-1} \abs{D_{\pi_{i}}^{(t)}} < \frac{L^{(t)}}{2}$.
Thus $\sum_{i=\pidx}^{k} \abs{D_{\pi_{i}}^{(t)}} = L^{(t)} - \sum_{i=1}^{\pidx-1} \abs{D_{\pi_{i}}^{(t)}} > \frac{L^{(t)}}{2}$.
Next, we make a case distinction and let $\tau = \MED(D_{\pi_{\pidx}}^{(t)})$. 

Case 1: $H^{\ell}_{\tau}$ has a left-perfect matching.
The algorithm records $\tau$ as the current best feasible result. 
All the distance values in $D^{(t)}$ that are equal or greater than $\tau$ are thus removed from the search space. 
Recall that $D_{\pi_i}^{(t)}$ are sorted in non-decreasing order according to the distance value of $\tau$. 
At least $\frac{1}{2} \sum_{i=\pidx}^{k} \abs{D_{\pi_{i}}^{(t)}} \geq \frac{L^{(t)}}{4}$ distance values are removed from the search space.

Case 2: $H^{\ell}_{\tau}$ has no left-perfect matching. 
All the distance values in $D^{(t)}$ that are less than or equal to $\tau$ are removed from the search space.
Thus at least $\frac{1}{2} \sum_{i=1}^{\pidx} \abs{D_{\pi_{i}}^{(t)}} \geq \frac{L^{(t)}}{4}$ distance values are removed from the search space.

Thus, the size of the search space is reduced by at least $\frac{1}{4}$.
\hfill\qed

\end{document}